\pdfoutput=1
\documentclass{article}
\usepackage[margin=1in]{geometry}
\usepackage{mathtools}
\usepackage{amssymb, amsthm, braket}
\usepackage{graphicx}
\usepackage{parskip}
\usepackage{authblk}
\RequirePackage[title,titletoc]{appendix}

\usepackage{hyperref}
\newtheorem{thm}{Theorem}[section]
\newtheorem{prop}[thm]{Proposition}
\newtheorem{lem}[thm]{Lemma}
\newtheorem{cor}[thm]{Corollary}

\theoremstyle{definition}
\newtheorem{dfn}[thm]{Definition}
\newtheorem{eg}[thm]{Example}

\theoremstyle{remark}
\newtheorem{rmk}[thm]{Remark}

\theoremstyle{definition}

\newcommand{\N}{\mathbb N}

\newcommand{\comm}[1]{}

\title{Flexible Catalysis}
\author{M\'at\'e Weisz\thanks{Corresponding author. Email: mate.weisz@linacre.ox.ac.uk.} }
\author{Sergii Strelchuk}

\affil{\small Department of Computer Science, University of Oxford, Parks Rd, Oxford OX1 3QG, United Kingdom}

\date{}

\begin{document}

\maketitle
\begin{abstract}
    In quantum information and computation, a central challenge is to determine which quantum states can be transformed into which others under restricted sets of free operations. While many transformations are impossible directly, catalytic processes can enable otherwise forbidden conversions: an auxiliary quantum state (the catalyst) facilitates the transformation while remaining unchanged. In this work, we introduce flexible catalysis, a generalization in which the catalyst is allowed to transform into a different auxiliary state, provided it remains a valid catalyst. We show that this framework subsumes both standard catalytic and multicopy transformations, and we analyse its advantages across several classes of free operations. In particular, we prove that when the free operations are local unitaries or permutation matrices, flexible catalysis enables state extractions that are unattainable with standard catalysis alone.
\end{abstract}

\section{Introduction}

Catalysis is a process where an auxiliary resource, known as a catalyst, enables a computational or information theoretic transformation that would otherwise be forbidden. This catalytic resource participates in the process but is returned to its original state upon completion, allowing it to be reused.

\vspace{0.2cm}

Jonathan and Plenio \cite{jonathan1999entanglement} were first to demonstrate that the presence of an ancillary entangled state could facilitate transformations between pure bipartite quantum states under Local Operations and Classical Communication (LOCC). They showed that for certain states $|\psi\rangle$ and $|\phi\rangle$, the direct transformation $|\psi\rangle \to |\phi\rangle$ 
is impossible via any LOCC protocol. However, by introducing a particular entangled state $|\eta\rangle$, the joint transformation $|\psi\rangle\otimes|\eta\rangle \to |\phi\rangle\otimes|\eta\rangle$ becomes feasible, with the catalyst $|\eta\rangle$ returned in its exact original form. This remarkable finding about the power of entanglement was aptly summarised in \cite{jonathan1999entanglement} as ``stranger kind of resource, one that can be used without being consumed at all".

\vspace{0.2cm}

Since the initial discovery, catalysis found numerous uses in  quantum information processing: from quantum resource theories, where it has emerged as a key mechanism for overcoming the restricted nature of so-called `free operations' \cite{lipka2025finite,chitambar2019quantum} to quantum thermodynamics, where catalysts can enable state conversions that would otherwise appear to violate thermodynamic laws for small-scale systems \cite{ng2015limits,brandao2015second}. While in general one has to check an infinite number of inequalities to determine the possibility of catalysis, recently, progress has been made in finding sufficient conditions which reduce the number of conditions to finitely many \cite{elkouss2025finite}. In the theory of coherence, for example, Aberg \cite{aaberg2014catalytic} showed that a highly coherent ancilla (such as a laser-like state) can serve as a catalyst for otherwise forbidden operations constrained by energy conservation. The coherence of this catalyst remains essentially undiminished, allowing ``latent" energy in superposition states to be released without depleting the resource. In quantum thermodynamics, catalysts (e.g. non-equilibrium auxiliary states) enable state transformations beyond the ordinary second-law constraints: any out-of-equilibrium state can in principle act as a universal catalyst for allowed thermodynamic transitions when sufficiently many copies are provided \cite{lipka2021all}. Likewise, in resource theories of asymmetry, catalysts can help bypass symmetry restrictions \cite{ding2021amplifying}. It is worth mentioning that catalysis is not all-powerful: there also exist fundamental limits. The authors of \cite{ding2021amplifying} prove a strict ``no-catalysis" theorem which forbids certain symmetry-protected transformations with any finite catalyst under perfectly symmetric operations. However, if one allows large or correlated catalysts, even an arbitrarily small asymmetry present in a state can be catalytically amplified to approach a maximally asymmetric state.

\vspace{0.2cm}

Under LOCC, multiple variations and special cases of catalysis were studied. Self-catalysis \cite{selfcatalysis} is the special case when the initial state $\ket{\psi}$ can be used a catalyst for a conversion $\ket
\psi\rightarrow\ket{\phi}$. Mutual catalysis \cite{mutualcatalysis} is the variant in which two transformations $\ket{\psi_1}\rightarrow \ket{\phi_1}$ and $\ket{\psi_2}\rightarrow \ket{\phi_2}$ that cannot be realised directly mutually catalyse each other, i.e. $\ket{\psi_1}\otimes \ket{\psi_2}\rightarrow \ket{\phi_1}\otimes \ket{\phi_2}$. Supercatalysis \cite{supercatalysis} is when the entanglement entropy of the auxiliary state increases while enabling a forbidden state transition. The notion of catalysis -- originally introduced for pure states -- has also been extended to mixed states \cite{MixedStateEntanglementCatalysis}. In the mixed state case, one may allow the catalyst to build up correlations with the main system \cite{CatQuantumTeleportation}.

\vspace{0.2cm}

In addition to these works, many other quantum-information-theoretic applications of catalysis were investigated. For a thorough overview of this line of research, we refer to \cite{lipka2024catalysis}.
 
\vspace{0.2cm}

Recent remarkable results highlight the utility of catalysis in quantum computing: in the context of fault-tolerant catalytic magic distillation, \cite{fang2024surpassing} introduced a catalytic approach for purifying noisy quantum states into magic states. Catalysis also appears in space-bounded computation, where the catalyst is an auxiliary memory register that must be restored exactly at the end of the computation, as in classical catalytic logspace \cite{buhrman2014computing} and its quantum analogue, quantum catalytic logspace (QCL) \cite{buhrman2025quantum}, a setting related to restricted models such as DQC1 (the so-called one-clean-qubit  model) \cite{knill1998power}. From a theoretical perspective, QCL exemplifies how the concept of catalysis can be repurposed to characterise reusable space in quantum complexity theory. It connects naturally with restricted models like DQC1 \cite{knill1998power} and highlights the operational significance of reversibility and restorability in quantum computation. Moreover, QCL raises foundational questions about the feasibility and limitations of exact catalyst recovery in practice. Allowing approximate or flexible recovery -- paralleling recent advances in flexible catalysis may yield even more expressive or realistic space-bounded models.

\vspace{0.2cm}

Before introducing the main idea of this paper, let us stop here to understand what makes catalysis really interesting and useful. Getting our catalyst back at the end of the process is impressive because it means that we have not consumed any more resource than during a non-catalytic transformation. In particular, if we want to repeat the transformation on a new object, we can use the same catalyst as before, without introducing any additional resources to the system. And after the second transformation, we can use it for a third time as well, and so on. In fact, most of the time this is all we care about: that adding a single catalyst can assist us with an arbitrary number of transformations. In other words, it is not the catalyst that we want to preserve for its own sake. It is the ability to perform the transformation again and again. So, even if the additional object we get back changes after every step, we should be happy as long as it can also act as a catalyst.

\vspace{0.2cm}

This suggests that our definition of catalysis can be extended. We would like to say something like ``a catalyst is an object that facilitates the transformation in a way that returns some other catalyst in the end". While this definition is, of course, circular, it is not difficult to come up with a sensible one that captures the idea. Thus, we arrive at the notion of flexible catalysis: having a set of objects (the set of catalysts), any element of which can be added to the initial object of the desired transformation, and together they can be transformed into the desired final object and another object from the set of catalysts.

\vspace{0.2cm}

Now, catalysis in the traditional sense is just a special case of our more general notion of flexible catalysis. It is natural to ask how much we have gained by relaxing the condition that we must get the original catalyst back exactly, or indeed, if we have gained anything at all. This paper aims to answer this question in some contexts within quantum information theory. In particular, we will show examples where flexibility does provide an advantage over traditional catalysis, and also other examples where it does not.

\vspace{0.2cm}

Since the basic idea of (flexible) catalysis is very simple and makes sense in any resource theory, we will be able to use some very general techniques to get an understanding of the basics of (flexible) catalytic transformations. Therefore, most of our initial statements will be about different mathematical resource theories, and then we will use these to obtain the corresponding results in quantum information.

\vspace{0.2cm}

A resource theory in quantum information (and in general) is only interesting if we do not have the ability to perform all imaginable transformations freely. So we will be considering restricted classes of quantum operations, and assume that we only have access to these. In particular, we will study catalysis under Local Operations and Classical Communication (LOCC), Local Unitaries (LU), and Permutation Matrices (PM). Informally, we can think of LOCC as the class of  operations that do not increase entanglement, LU as the class of entanglement-preserving operations (or as the reversible operations of LOCC), and PM as a copy of reversible classical operations embedded in the set of unitary quantum gates.

\vspace{0.2cm}

The rest of this paper is structured as follows. In Section \ref{Informal statements of main results}, we present the informal statements of our main results. In Section \ref{Transformation theories}, we introduce the mathematical language we need to talk about catalysis in resource theories in general, and we characterise the quantum information-theoretic resource theories that we are interested in. We also give the definitions of the different versions of catalysis that we will use in the rest of the paper, and we prove some elementary results about them that hold in any resource theory. In Section \ref{Flexible catalysis of multisets}, we focus on a class of abstract mathematical resource theories, motivated by our results in Section \ref{Transformation theories}, and we collect results addressing the power of flexible catalysis. This is where most of the technical work is done. In Section \ref{Flexible catalysis in quantum information}, we give the formal statements and proofs of our main results. In Section \ref{Flexible catalysis outperforms catalytic multicopy extractions}, we prove an additional abstract result about the power of flexible catalysis. Finally, in Section \ref{Open questions}, we list some interesting open questions.

\section{Informal statements of main results}
\label{Informal statements of main results}

In this section we present a selection of our results that demonstrates the power of flexible catalysis in quantum information theory. We only state them informally here, together with brief explanations of their significance, but without any proofs. We do not give any precise definitions of the terms used in the statements either, only draw attention to the following. In many naturally arising resource theories (e.g. LOCC), discarding objects is a free operation. However, some classes of quantum operations (e.g. LU) do not include a discard operation, which would allow us to keep a factor of a tensor product of quantum states, and discard the other. For these operation classes, it makes sense to talk about transformations that also allow discard operations alongside the operations from the given class. We shall call these transformations extractions. It turns out that, when it comes to catalysis, there are significant differences between extractions and transformations without discard operations. Therefore, we will always make it clear if discard operations are also allowed, provided that they are not already included in the class of quantum operations.

\begin{dfn} (Informal version of Definition \ref{precise_defn_of_flex_cat}.)
        We say that there is a flexible catalytic conversion from $\ket{\psi}$ to $\ket{\phi}$ if there exists a set $S$ of quantum states (called the catalytic set, or set of catalysts) such that for any catalyst $\ket{\eta} \in S$, there exists another catalyst $\ket{\eta'}\in S$ for which $\ket{\psi}\otimes\ket{\eta}$ can be converted into $\ket{\phi}\otimes\ket{\eta'}$ under the set of free operations.
\end{dfn}

\begin{prop} (Informal version of Proposition \ref{LOCCpositive}.)
\label{LOCC_main_result}
    There exist bipartite quantum states $\ket{\psi},\ket{\phi}$ such that preparing $\ket{\phi}$ from $\ket{\psi}$ with certainty

    (i) is possible by a flexibly catalytic $\mathrm{LOCC}$ procedure $P$,

    (ii) but no single catalyst used in $P$ can do the job alone.
\end{prop}

This result means that in LOCC we can do more with the same catalyst and initial states if we are willing to be flexible than if we restrict ourselves to traditional catalysis. However, it does not say that flexible catalysis opens up new transformations that are not possible via traditional catalysis at all.

\begin{thm} (Informal version of Theorem \ref{CatExtLUfin_beats_CatExtLU}.)
\label{LU_main_result}
    There exist bipartite quantum states $\ket{\psi},\ket{\phi}$ such that extracting $\ket{\phi}$ from $\ket{\psi}$

    (i) cannot be realised by any catalytic $\mathrm{LU}$ procedure,

    (ii) but can be realised by a flexible catalytic $\mathrm{LU}$ procedure.

\end{thm}

\begin{thm} (Informal version of Theorem \ref{flexhelp}.)
\label{PM_main_result}
    There exist quantum states $\ket{\psi},\ket{\phi}$ such that extracting $\ket{\phi}$ from $\ket{\psi}$

    (i) cannot be realised catalytically with permutation matrices,

    (ii) but can be realised by a flexible catalytic procedure using permutation matrices.
\end{thm}

Theorems \ref{LU_main_result} and \ref{PM_main_result} mean that more extractions can be done by flexible catalysis than by regular catalysis, if the class of free operations is either local unitaries or permutation matrices. So this is a stronger advantage provided by flexibility than that of Theorem \ref{LOCC_main_result} for LOCC.

\vspace{0.2cm}

The following result holds for any resource theory. It aims to give the reader a sense of the strength of flexible catalysis in general, relating it to well-studied types of transformations. 

\begin{thm} (Informal version of Proposition \ref{flexvsmulticopyextrtt}.)
    Flexible catalysis can be regarded as a generalisation of catalytic multicopy transformations.
\end{thm}

In particular, the class of flexible catalytic transformations contains all multicopy transformations and all catalytic transformations. Since these two transformation classes are closely related in many transformation theories \cite{duan2005multiple,lipka2024catalysis}, it is enlightening to view them as special cases of a more general class.

\vspace{0.2cm}

The last statement we list in this section is not directly about flexible catalysis. However, it addresses a question that arises naturally when we investigate entanglement catalysis, and we found it interesting enough on its own to include it here.

\begin{thm} (Informal version of Theorem \ref{nouniqfact}.)
\label{there_is_no_unique_factorisation_of_bipartite_entanglement}
    There is no unique factorisation of bipartite entanglement classes.
\end{thm}

Here, by a bipartite entanglement class we mean an LU equivalence class of bipartite quantum states. So Theorem \ref{there_is_no_unique_factorisation_of_bipartite_entanglement} says that knowing the LU equivalence class of a tensor product of bipartite states is not sufficient to determine the LU equivalence classes of the individual factors. 

\section{Transformation theories}
\label{Transformation theories}

We start this section by giving the precise definition of the mathematical structures we will use to capture resource theories. We will call these structures transformation theories (TT). We develop the language required to say that two TTs are essentially the same, so that we can apply results proved for one TT to the other. Then we define classes of transformations and extractions assisted by catalysis or flexible catalysis, and prove some elementary results about them that hold regardless of what TT we are working in. Finally, we give the mathematical characterisation of the quantum-information-theoretic TTs of LOCC, LU, and PM.

\subsection{Basics}

\begin{dfn} (Transformation theory.) A transformation theory is a triple $(M, +, \leq )$, where $(M,+)$ is a monoid and $\leq$ is a translation-invariant preorder on $M$ satisfying $a+b\leq b+a$ for all $a,b\in M$. We say that a transformation theory is symmetric if $a\leq b \implies b\leq a$ for all $a,b\in M$, i.e. if $\leq$ is an equivalence relation on $M$.
\end{dfn}

\begin{rmk}
    We call the elements of $M$ the states of the transformation theory. If $a\leq b$ for some states $a,b\in M$, we say that $a$ transforms into $b$. We can think of the monoid operation $+$ as the rule for composing states.
\end{rmk}

\begin{rmk}
    The condition $a+b\leq b+a$ for all $a,b\in M$ has the intuitive meaning that swapping two states is allowed in any transformation theory. There exist resource theories in which swapping states is not free, but we will not work with any such resource theories in this paper, so we can incorporate this assumption into the definition of our mathematical model.
\end{rmk}

\begin{rmk}
    We will often abuse notation in the usual way, and identify a transformation theory with the underlying monoid or set.
\end{rmk}

\begin{rmk}
    This definition is equivalent to that of a theory of resource convertibility in \cite{MathematicalTheoryOfResources}. It is also similar to the definition of a resource theory in \cite{FRITZ_2015}: resource theories are ordered commutative monoids. The only difference is that our transformation theories do not require the preorder $\leq$ to be antisymmetric, i.e. mutually interconvertible states are not forced to be equal. In particular, any ordered commutative monoid is a TT. We find it convenient to work with this definition, but everything we do could be also be phrased using the definition of \cite{FRITZ_2015}. Conversely, anything modelled by ordered commutative monoids can also be phrased in the language of TTs. From this point, we will only use the term resource theory in an informal way, except when referring to the resource theory definition of \cite{FRITZ_2015}. We will use transformation theory (or TT) to refer to the mathematical structure defined above. See also Remark \ref{resource_theories_and_quotients_of_TTs}.
\end{rmk}

  \begin{eg}
\label{TT's from classes quantum operations}
  
      Let $S$ be the set of pure quantum states living in finite dimensional Hilbert spaces, and let $\mathcal{C}$ be a set of quantum operations that includes all permutations of subsystems and is closed under parallel and sequential composition. For $\ket{\psi}, \ket{\phi}\in S$, write $\ket{\psi}\rightarrow_\mathcal{C} \ket{\phi}$ iff there is an operation $C \in \mathcal{C}$ sending $\ket{\psi}\bra{\psi}$ to $\ket{\phi}\bra{\phi}.$ Then $(S,\otimes, \rightarrow_\mathcal{C})$ is a transformation theory. Moreover, if all elements of $\mathcal{C}$ are invertible and the inverses are also in $\mathcal{C}$, then $(S,\otimes, \rightarrow_\mathcal{C})$ is symmetric.
  \end{eg}

\begin{prop}
\label{orderdoesntmatter}
    Let $T$ be a transformation theory, $\sigma \in S_n$ a permutation, and $a_1,\ldots,a_n\in T$. Then $a_1+\ldots +a_n\leq a_{\sigma(1)}+\ldots+a_{\sigma(n)}.$
\end{prop}
\begin{proof}
    Any permutation can be written as a composition of swaps, so this follows immediately from the assumption that $a+b\leq b+a$ for all $a,b\in T$. 
\end{proof}

\begin{dfn} (Morphism of transformation theories.)
    Let $T=(M,+_T,\leq_T)$ and $S=(N,+_S,\leq_S)$ be transformation theories. A morphism of transformation theories $\phi: T\rightarrow S$ is a map $M\rightarrow N$ satisfying $\phi(a)+\phi(b)\sim_S \phi(a+b)$ and $a\leq_T b \implies \phi(a)\leq_S \phi(b)$ for all $a,b\in M$. An isomorphism is a morphism which has a two-sided inverse.
\end{dfn}

\begin{rmk}
    $\sim$ denotes equivalence of states: $a\sim b\iff a\leq b$ and $b\leq a$. This is an equivalence relation, since $\leq$ is a preorder.
\end{rmk}

\begin{dfn} (Equivalence of transformation theories.) Let $T=(M,+_T,\leq_T)$ and $S=(N,+_S,\leq_S)$ be transformation theories. An equivalence  $\phi: T\rightarrow S$ is a morphism of transformation theories that is additionally essentially surjective (meaning that for all $a\in N \ \exists \ b\in M$ such that $a\sim \phi(b)$) and satisfies $a\leq_T b \impliedby \phi(a)\leq_S \phi(b)$.
\end{dfn}

\begin{prop}
    \label{addingdiscardop} (Transformation theory of extractions.) Let $T=(M, +, \leq_T)$ be a transformation theory. Define the relation $\leq_{T'}$ as follows: $A\leq_{T'} B \iff \exists D\in M: A\leq_T B+D$. Then $T'=(M, +, \leq_{T'})$ is a transformation theory.
\end{prop}
\begin{proof}
    Straightforward.
\end{proof}

\begin{rmk}
    $T'$ is the transformation obtained from $T$ by adding a free discard operation. If $T$ already had the discard operation, then $T=T'$. It will be convenient to talk about transformations in $T'$ without mentioning $T'$ explicitly. Instead, we can call them extractions under $T$.
\end{rmk}

\begin{prop} (Transformation theory of equivalence classes.) Let $T=(M,+, \leq)$ be a transformation theory. Then $T/\sim := (M/\sim , +, \leq )$ is also a transformation theory, where $+, \leq$ are defined on the equivalence classes in the natural way. Moreover, it is equivalent to the original transformation theory $T$.
\end{prop}
\begin{proof}
    Straightforward.
\end{proof}

\begin{rmk}
\label{resource_theories_and_quotients_of_TTs}
Note that such a transformation theory of equivalence classes is an ordered commutative monoid, so it has the structure of a resource theory, according to the definition of \cite{FRITZ_2015}.
\end{rmk}

\begin{prop}
\label{eqtt}
    Let  $T=(M,+_T, \leq_T)$ and  $S=(N,+_S, \leq_S)$ be transformation theories. The following are equivalent:
    
    (i) $T$ and $S$ are equivalent,
    
    (ii) $T/\sim_T$ and $S/\sim_S$ are equivalent,
    
    (iii) $T/\sim_T$ and $S/\sim_S$ are isomorphic.

\end{prop}
\begin{proof}
    (iii) $\implies$ (ii): Clear.

    (ii) $\implies$ (i): If $\Phi: T/\sim_T \rightarrow S/\sim_S$ is an equivalence, let $\phi : T\rightarrow S$ be the map defined by the composition $t\mapsto [t]\xmapsto{\Phi}[s]\mapsto s$, where the last map chooses an arbitrary representative of the equivalence class. This is a composition of three equivalences, thus also an equivalence.

    (i) $\implies$ (iii): If $\phi : T\rightarrow S$ is an equivalence, it induces a map $\Phi: T/\sim_T \rightarrow S/\sim_S$. It is easy to check that $\Phi$ is a bijective morphism, monotonic in both directions. So it has an inverse, and thus it's an isomorphism.
\end{proof}

\begin{rmk}

    In particular, equivalence of transformation theories is a symmetric relation. It is clearly reflexive and transitive, so it is also an equivalence relation.

    \vspace{0.2cm}

    Proposition \ref{eqtt} shows that equivalent transformation theories can be considered the same for all practical purposes. In the rest of this paper, we will copy results from one TT to another equivalent TT without further explanation. 
\end{rmk}

\subsection{Transformation classes}

Here we introduce the catalytic transformation classes that will allow us to formalise our statements about the strengths of certain types of catalytic transformations. In general, a catalytic transformation class is the set of ordered pairs $(A,B)$ of states, such that $B$ can be prepared or extracted from $A$ catalytically, with some restriction on exactly what kind of catalysis we allow. For a transformation theory $T$, write $\mathrm{Tr}_T:=\{(A,B)\in T\times T: A\leq B\}$ for the class of direct conversions.

\begin{dfn} (Catalytic transformations.) For any transformation theory $T$ and a state $C\in T$, define
    \begin{equation}
        \mathrm{Cat}_T(C):=\{(A,B)\in T\times T: A+C\leq B+C\},
    \end{equation}
    and let
    \begin{equation}
        \mathrm{Cat}_T :=\bigcup_{C\in T} \mathrm{Cat}_T(C).
    \end{equation}

For $(A,B)\in \mathrm{Cat}_T$, we say that $A$ catalytically transforms into $B$.
\end{dfn}

\begin{rmk}
    If $T$ is a TT of the form $(S,\otimes,\rightarrow_{\mathcal{C}})$ (using the notation of Example \ref{TT's from classes quantum operations}), then we recover the standard notion of catalysis in quantum information: $\ket{\psi}$ catalytically transforms into $\ket{\phi}$ (i.e. $(\ket{\psi},\ket{\phi})\in \mathrm{Cat}_T$) if there exists $\ket{\eta}$ for which $\ket{\psi}\otimes\ket{\eta}\rightarrow_\mathcal{C}\ket{\phi}\otimes \ket{\eta}$. If these TTs also included mixed states, it would also be possible to allow correlations between the target state and the catalyst, and thus we would recover the notion of correlated catalysis. However, in this paper we will stick to the stricter definition of catalysis.
\end{rmk}

\begin{dfn}
\label{precise_defn_of_flex_cat}
(Flexible catalytic transformations.) Let $S_T$ be the collection of all nonempty subsets of $T$. For any $S\in S_T$ (i.e. any nonempty subset of $T$), define
    \begin{equation}
        \mathrm{Cat}_T^{(f)}(S):=\{(A,B)\in T\times T: \forall C\in S \ \exists C' \in S \text{ such that } A+C\leq B+C'\},
    \end{equation}
    and
    \begin{equation}
        \mathrm{Cat}_T^{(f)}:=\bigcup_{S\in S_T} \mathrm{Cat}_T^{(f)}(S).
    \end{equation}

For $(A,B)\in \mathrm{Cat}_T^{(f)}$, we say that there is a flexible catalytic transformation taking $A$ to $B$. If $(A,B)\in \mathrm{Cat}_T^{(f)}(S)$, we say that $S$ is a set of catalysts (or catalytic set) for this flexible catalytic transformation.

  Also define
  \begin{equation}
      \mathrm{Cat}_T^{(\mathrm{fin})}:=\bigcup_{S\in S_T, |S|<\infty} \mathrm{Cat}_T^{(f)}(S),
  \end{equation}

so that $(A,B)\in \mathrm{Cat}_T^{(\mathrm{fin})}$ if and only if there is a finite catalytic set $S$ that allows a flexible catalytic conversion from $A$ to $B$. Let $\mathrm{Cat}_T^{(\leq n)}\subseteq \mathrm{Cat}_T^{(\mathrm{fin})}$ be the the set of those pairs $(A,B)$ for which $A$ can be converted to $B$ using a set of at most $n$ catalysts, i.e.
  \begin{equation}
      \mathrm{Cat}_T^{(\leq n)}:=\bigcup_{S\in S_T, |S|\leq n} \mathrm{Cat}_T^{(f)}(S).
  \end{equation}

  Finally, define $\mathrm{Cat}_T^{(n)}$ to be the class corresponding to flexible catalytic transformations that can be realised by cyclically permuting exactly $n$ catalysts, i.e.

  \begin{equation*}
        \mathrm{Cat}_T^{(n)}:=\{(A,B)\in T\times T: \exists C_0, C_1, \ldots, C_n \text{ with } C_n = C_0
  \end{equation*}
   \begin{equation}
       \text{ such that } A+C_{i-1}\leq B+C_i \ \forall i\in \{1,2,\ldots,n\}\}.
  \end{equation}

  \end{dfn}

   \begin{dfn} (Extractions.)  Let $T$ be a transformation theory and let $T'$ be its extension with free discard operations, as defined in Proposition \ref{addingdiscardop}. Define
   
  \begin{equation}
      \mathrm{Ext}_T:=\mathrm{Tr}_{T'}.
  \end{equation}

  For $(A,B)\in \mathrm{Ext}_T$, we say $B$ is extractable from $A$. Similarly, we define the extraction version of the above transformation classes by adding free discard operations.

  For any $C\in T$, define
    \begin{equation}
      \mathrm{CatExt}_T(C):=\mathrm{Cat}_{T'}(C)
  \end{equation}

  and 
  \begin{equation}
        \mathrm{CatExt}_T:=\mathrm{Cat}_{T'}.
    \end{equation}

  For $(A,B)\in \mathrm{CatExt}_T$, we say that $B$ is catalytically extractable from $A$.

  For any $S\in S_T$, define
    \begin{equation}
        \mathrm{CatExt}_T^{(f)}(S):= \mathrm{Cat}_{T'}^{(f)}(S)
    \end{equation}
    
    and
    \begin{equation}
        \mathrm{CatExt}_T^{(f)}:=\mathrm{Cat}_{T'}^{(f)}.
    \end{equation}

For $(A,B)\in \mathrm{CatExt}_T^{(f)}$, we say that $B$ is flexibly catalytically extractable from $A$.
 
    Also define 
  \begin{equation}
        \mathrm{CatExt}_T^{(\mathrm{fin})}:= \mathrm{Cat}_{T'}^{(\mathrm{fin})},
    \end{equation}
     \begin{equation}
        \mathrm{CatExt}_T^{(\leq n)}:=\mathrm{Cat}_{T'}^{(\leq n)},
    \end{equation}

  \begin{equation}
        \mathrm{CatExt}_T^{(n)}:=  \mathrm{Cat}_{T'}^{(n)}.
  \end{equation}

  \end{dfn}

\begin{rmk}
    We introduced this additional notation to emphasise the intuition that extraction is just a generalised transformation in the original TT, even though it is possible to think of it as a transformation in a different TT.
\end{rmk}

\begin{prop}
\label{Union_of_cyclic_classes}
      \begin{equation}
          \mathrm{Cat}_T^{(\leq n)}=\bigcup_{1\leq k\leq n} \mathrm{Cat}_T^{(k)}.
      \end{equation}
\end{prop}
\begin{proof}

By definition, we have
    \begin{equation}
        \mathrm{Cat}_T^{(\leq n)}\supseteq \bigcup_{1\leq k\leq n} \mathrm{Cat}_T^{(k)}.
    \end{equation}
    For the reverse inclusion, assume that $(A,B)\in \mathrm{Cat}_T^{(f)}(S)$ where $|S|\leq n$. Note that if we draw a directed graph $G$ with vertex set $S$ and draw an edge from $C\in S$ to $D\in S$ iff $A+C\leq B+D$, then (as there is an edge starting at every vertex) $G$ contains a directed cycle. The vertices of this cycle give us a set of $k\leq n$ catalysts whose cyclic permutation realises a flexible catalytic conversion from $A$ to $B$, so $(A,B)\in \mathrm{Cat}_T^{(k)}$.
\end{proof}

 \begin{prop}
 \label{infeqtrcat}
      Let $T$ be any transformation theory. Then $\mathrm{CatExt}_T^{(f)}=\mathrm{Cat}_T^{(f)}$.
  \end{prop}
  
  \begin{proof}
      One inclusion is trivial, so we only need to show $\mathrm{CatExt}_T^{(f)}\subseteq \mathrm{Cat}_T^{(f)}$. Take $(A,B)\in \mathrm{CatExt}_T^{(f)}$, and let $S\in S_T$ be such that $(A,B)\in \mathrm{CatExt}_T^{(f)}(S)$. We'll define a chain of catalysts in $S$ inductively: let $C_0\in S$ be arbitrary. Given $C_{n-1}$, let $C_n\in S,D_n\in T$ be chosen such that $A+C_{n-1}\leq B+D_n+C_n$. But now we can set $C_n'=C_n+(D_1+\ldots+D_n)$ for all $n\in \mathbb{N}.$ Then we have $A+C_{n-1}'\leq B+C_n'$ for all $n\geq 1$, and thus $(A,B)\in \mathrm{Cat}_T^{(f)}(S')$ with $S'=\{C_0',C_1',\ldots\}$.
  \end{proof}
  
\begin{center}
    
  \includegraphics[width=10cm]{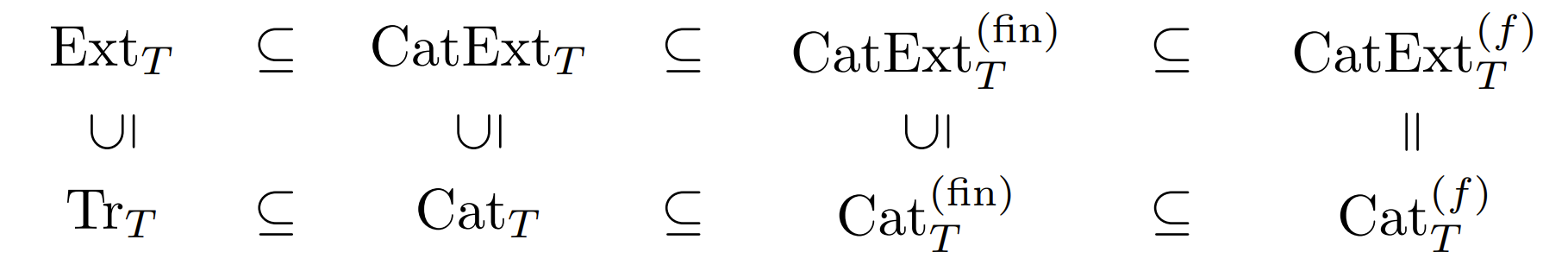}
    
    \textbf{Figure 1}: Inclusion diagram of catalytic transformation classes for a general transformation theory.
\end{center}

\begin{prop}
\label{flexvsmulticop}
    Let $T$ be a transformation theory, and let $n$ be a positive integer. Then $(A,B)\in \mathrm{Cat}_T^{(n)}$ if and only if $(nA,nB)\in \mathrm{Cat}_T$.
\end{prop}

\begin{proof}
    If $A+C_{m-1}\leq B+C_m$ for all $1\leq m \leq n$ (where $C_0=C_n$), then by adding all of these together and using Proposition \ref{orderdoesntmatter}, we get
    \begin{equation}
            nA+(C_1+\ldots+C_n)\leq nB+(C_1+\ldots+C_n).
    \end{equation}
    Conversely, if $nA+C\leq nB+C$, then we have \begin{equation}
        A+((n-1)A+C)\leq B+((n-1)B+C)
    \end{equation}
    \begin{equation}
        A+((n-1)B+C)\leq B+(A+(n-2)B+C)
    \end{equation}
    \begin{equation}
        A+(A+(n-2)B+C)\leq B+(2A+(n-3)B+C)
    \end{equation}
    $$\vdots$$
    \begin{equation}
        A+((n-2)A+B+C)\leq B+((n-1)A+C).
    \end{equation}
\end{proof}

\begin{rmk}
    This shows that catalytic multicopy transformations are a special case of flexible catalytic transformations. By Proposition \ref{Union_of_cyclic_classes}, for finite flexible catalysis we have a sort of converse as well, but notice that the space requirements of the  catalytic multicopy transformation derived from the flexible catalytic transformation are larger by a factor of $n$. Therefore, finite flexible catalytic transformations can be more practical than the corresponding catalytic multicopy transformations. In Example \ref{advantage} we will demonstrate a flexible catalytic transformation in which the catalysts are smaller in size than the initial state $A$, and which cannot be replaced by a catalytic protocol, and thus it beats any catalytic multicopy transformation for the same task.

    Note also that our construction of catalysts resembles the construction of \cite{duan2005multiple} for a single catalyst that can be used to simulate a multiple copy transformation by catalysis. The main difference is that our construction creates a finite set of catalysts that work together to achieve the same goal, rather than a single catalyst. Our result is also more general, since it holds for an arbitrary transformation theory.

    While the proof is quite simple, this result suggests that flexible catalysis might in some sense be the ``right" notion of catalysis, since it includes multiple copy transformations, which are incomparable to catalytic transformations in general.
\end{rmk}

\begin{prop}
\label{flexvsmulticopyextrtt}
    Let $T$ be a transformation theory, and let $n$ be a positive integer. Then $(A,B)\in \mathrm{CatExt}_T^{(n)}$ if and only if $(nA,nB)\in \mathrm{CatExt}_T$.
\end{prop}
\begin{proof}
  Apply Proposition \ref{flexvsmulticop} to the transformation theory $T'$ defined in Proposition \ref{addingdiscardop}.
\end{proof}

\subsection{The transformation theories of LOCC, LU and PM}
\label{The transformation theories of LOCC, LU and PM}

\begin{dfn}
    (Bipartite quantum state.) A bipartite quantum state is a triple $(\ket{\psi}, \mathcal{H}_A, \mathcal{H}_B)$, where $(\mathcal{H}_A, \mathcal{H}_B)$ are finite dimensional complex Hilbert spaces, and $\ket{\psi}\in \mathcal{H}_A\otimes \mathcal{H}_B$ is a normalised vector up to global phase. We will denote the set of all bipartite quantum states by $\mathcal{B}$.
\end{dfn}
\begin{rmk}
    We can refer to a bipartite quantum state by just the state vector $\ket{\psi}$, if the Hilbert spaces $\mathcal{H}_A, \mathcal{H}_B$ are clear from the context, or if it does not matter what they are. Note also that a finite dimensional complex Hilbert space is determined by its dimension (up to isomorphism).
\end{rmk}

\begin{dfn}
    (Swapping subsystems.)   For Hilbert spaces $\mathcal{H}_A, \mathcal{H}_B$, define the permutation matrix $\mathrm{SWAP}_{A,B}$ to be the unitary defined by
    $$\mathrm{SWAP}_{A,B}(\ket{\psi}_A\otimes\ket{\phi}_B)=\ket{\phi}_B\otimes\ket{\psi}_A$$
    for all $\ket{\psi}_A \in\mathcal{H}_A$, $\ket{\phi}_B \in\mathcal{H}_B$.
\end{dfn}

\begin{prop}
    (Monoid of bipartite quantum states.)  The set $\mathcal{B}$ of bipartite quantum states forms a monoid under the following operation, which we will simply denote as a tensor product:
    \begin{equation*}
    (\ket{\psi}, \mathcal{H}_A, \mathcal{H}_B)\otimes (\ket{\phi}, \mathcal{H}_C, \mathcal{H}_D)=
    \end{equation*}
    \begin{equation}
    =((I_A\otimes \mathrm{SWAP}_{B,C}\otimes I_D)(\ket{\psi}\otimes \ket{\phi}), \mathcal{H}_A\otimes \mathcal{H}_C, \mathcal{H}_B\otimes \mathcal{H}_D).
    \end{equation}
\end{prop}
\begin{proof}
    Straightforward.
\end{proof}
\begin{dfn}
    (Bipartite quantum operation.) 
    
    A bipartite quantum operation is a quintuple $(C, \mathcal{H}_A, \mathcal{H}_B,  \mathcal{H}_X, \mathcal{H}_Y)$ where  $\mathcal{H}_A, \mathcal{H}_B,\mathcal{H}_X, \mathcal{H}_Y$ are finite dimensional complex Hilbert spaces and $C: L(\mathcal{H}_A\otimes \mathcal{H}_B)\rightarrow L(\mathcal{H}_X\otimes \mathcal{H}_Y)$ is a quantum operation (i.e. a completely positive trace non-increasing map from the operator space $L(\mathcal{H}_A\otimes \mathcal{H}_B)$ the the operator space $L(\mathcal{H}_X\otimes \mathcal{H}_Y)$).
\end{dfn}
\begin{dfn}
    (LOCC \cite{bennett1996purification}.)  A bipartite quantum choperation $(C, \mathcal{H}_A, \mathcal{H}_B,  \mathcal{H}_X, \mathcal{H}_Y)$ is in LOCC if the quantum operation $C: L(\mathcal{H}_A\otimes \mathcal{H}_B)\rightarrow L(\mathcal{H}_X\otimes \mathcal{H}_Y)$ can be implemented by local quantum operations (i.e. quantum operations of the form $C_{A\rightarrow X}: L(\mathcal{H}_A)\rightarrow L(\mathcal{H}_X)$, $C_{B\rightarrow Y}: L(\mathcal{H}_B)\rightarrow L(\mathcal{H}_Y)$) and classical communication.
\end{dfn}
\begin{dfn}
\label{defn_of_lu}
    (LU.)  A bipartite quantum operation $(C, \mathcal{H}_A, \mathcal{H}_B,  \mathcal{H}_X, \mathcal{H}_Y)$ is in LU if $C$ is of the form $U_{A\rightarrow X}\otimes U_{B\rightarrow Y}$, where $U_{A\rightarrow X}: L(\mathcal{H}_A)\rightarrow L(\mathcal{H}_X) , U_{B\rightarrow Y}: L(\mathcal{H}_B)\rightarrow L(\mathcal{H}_Y)$ are partial isometries.
\end{dfn}
\begin{rmk}
    The reason we define LU using partial isometries rather than unitaries is to ensure that states with the same entanglement structure are mutually interconvertible even if the ambient Hilbert spaces have different dimensions. See also Remark \ref{remark_on_defn_of_lu_and_schmidt}.
\end{rmk}

The following two propositions are completely straightforward -- they are essentially special cases of Example \ref{TT's from classes quantum operations} -- so we present them without proof.

\begin{prop}
    (Transformation theory of LOCC.) For bipartite quantum states 
        $(\ket{\psi}, \mathcal{H}_A, \mathcal{H}_B),(\ket{\phi}, \mathcal{H}_X, \\ \mathcal{H}_Y)$, write $\ket{\psi}\rightarrow_{\mathrm{LOCC}}\ket{\phi}$ if there is a bipartite operation $(C, \mathcal{H}_A, \mathcal{H}_B,  \mathcal{H}_X, \mathcal{H}_Y)$ in LOCC sending $\ket{\psi}$ to $\ket{\phi}$ (or more precisely, $\ket{\psi}\bra{\psi}$ to $\ket{\phi}\bra{\phi}$). Then $(\mathcal{B}, \otimes, \rightarrow_{\mathrm{LOCC}})$ is a transformation theory.
\end{prop}

\begin{prop}
    (Transformation theory of LU.) For bipartite quantum states $(\ket{\psi}, \mathcal{H}_A, \mathcal{H}_B),(\ket{\phi}, \mathcal{H}_X, \\ \mathcal{H}_Y)$, write $\ket{\psi}\rightarrow_{\mathrm{LU}}\ket{\phi}$ if there is a bipartite operation $(C, \mathcal{H}_A, \mathcal{H}_B,  \mathcal{H}_X, \mathcal{H}_Y)$ in LU sending $\ket{\psi}$ to $\ket{\phi}$. Then $(\mathcal{B}, \otimes, \rightarrow_{\mathrm{LU}})$ is a symmetric transformation theory.
\end{prop}

To characterise these TT's, we will use the following well-known results (see \cite{nielsen2010quantum}).

\begin{thm} (Schmidt decomposition.)
    Let  $(\ket{\psi}, \mathcal{H}_A, \mathcal{H}_B)$ be a bipartite quantum state and let $d=\min \{\dim(\mathcal{H}_A), \dim(\mathcal{H}_B)\}$. Then there exist orthonormal sets $\{\ket{\alpha_1},\ldots,\ket{\alpha_d}\}\subset \mathcal{H}_A$, $\{\ket{\beta_1},\ldots,\ket{\beta_d}\}\subset \mathcal{H}_B$, and nonnegative real scalars $c_1,\ldots, c_d$ such that
    \begin{equation}
        \ket{\psi}=c_1\ket{\alpha_1}\ket{\beta_1}+\ldots+c_d\ket{\alpha_d}\ket{\beta_d}.
    \end{equation}
    Moreover, the scalars $c_1,\ldots, c_d$ (called the Schmidt coefficients) are unique up to reordering.
\end{thm}

\begin{thm}
\label{LOCCMAJ}
   Let $(\ket{\psi}, \mathcal{H}_A, \mathcal{H}_B),(\ket{\phi}, \mathcal{H}_X, \mathcal{H}_Y)$ be bipartite quantum states with squared Schmidt coefficients given by the vectors $u,v$, respectively. Then
    
    (i) there exists a bipartite $\mathrm{LOCC}$ operation sending $\ket{\psi}$ to $\ket{\phi}$ if and only if $u\prec v$;

    (ii) there exists a bipartite $\mathrm{LU}$ operation sending $\ket{\psi}$ to $\ket{\phi}$ if and only if $u\sim v$.
\end{thm}

\begin{rmk}
\label{remark_on_defn_of_lu_and_schmidt}
    Although this is often stated for $\mathcal{H}_A=\mathcal{H}_X$ and $\mathcal{H}_B=\mathcal{H}_Y$, and with LU consisting only of local unitaries, it is straightforward to extend the result to this more general setting. (The reason being that if a partial isometry preserves the normalisation of a state vector, then this vector must lie in the subspace on which the partial isometry acts as an isometry.)
\end{rmk}

\begin{rmk}
    A remark on the meanings of the symbols we use to denote different preorders on sets of vectors or multisets.
    
    Symbol $\prec$ denotes majorization. Although originally defined for vectors of equal length, it can be extended to vectors of different lengths by padding the shorter vector with 0's, provided that all entries are nonnegative reals. Moreover, we can extend this definition to multisets, since majorization does not depend on the order of the elements.

    Symbol $\sim$ denotes equality up to reordering and adding or removing 0's. It is easy to see that $u\sim v \iff u\prec v$ and $v\prec u$. So $\sim$ is an equivalence relation.

    $\propto$ will denote equality up to constant translation by the group operation.
\end{rmk}

\begin{prop}
\label{description of TTs of LOCC and LU}
Let $V_0$ be the set of discrete probability distribution vectors up to permutation. (Equivalently, $V_0$ is the set of finite multisets of nonnegative reals summing up to 1.)
Then
   
    (i) The transformation theory of $\mathrm{LOCC}$ is equivalent to the transformation theory $\mathrm{MAJ}_0=(V_0, \otimes, \prec)$.

    (ii) The transformation theory of $\mathrm{LU}$ is equivalent to the transformation theory $\mathrm{EQ}_0=(V_0, \otimes, \sim)$.

\end{prop}
\begin{proof}
    Let $\gamma: \mathcal{B}\rightarrow V_0$ be the map sending any $(\ket{\psi}, \mathcal{H}_A, \mathcal{H}_B)\in \mathcal{B}$ to the multiset given by the squares of the Schmidt coefficients. Clearly $\gamma$ is surjective. It is easy to see that $\gamma$ is a monoid homomorphism, and thus, by Theorem \ref{LOCCMAJ}, $\gamma$ gives rise to both of the claimed equivalences.
\end{proof}

\begin{prop}
   \label{chain_of_equivalences}
    Let $V_1$ be the set of finite multisets of positive reals summing up to 1, $V$ the set of finite, nonempty multisets of positive reals, and $M_\mathbb{R}$ the set of finite, nonempty multisets of reals. Then the transformation theories  $\mathrm{EQ}_0$, $\mathrm{EQ}_1:=(V_1, \otimes, =)$, $\mathrm{EQ}_p:=(V, \otimes, \propto)$, $T_\mathbb{R}':=(M_\mathbb{R}, +, \propto)$ are all equivalent. (Addition of multisets is defined in the natural way. See Definition \ref{defn of addition of multisets}.)
    
\end{prop}

\begin{proof}
    Let $f_1:V_0\rightarrow V_1$ be the map that deletes any 0's from all multisets, $f_2:V\rightarrow V_1$ the map that normalises the sum of any multiset, $f_3:M_\mathbb{R}\rightarrow V$ the elementwise exponential map. It is straightforward to check that these maps give rise to the claimed equivalences.
\end{proof}

\begin{cor}
\label{ttlueq}
    The transformation theory of $\mathrm{LU}$ is equivalent to $T_\mathbb{R}'.$
\end{cor}

\begin{rmk}
\label{ttofpm}
    Let $\mathcal{P}$ denote the set of permutation unitaries, and let $S$ be the set of quantum state vectors not containing any 0 entries. Then it is easy to see that the transformation theory $\mathrm{PM}=(S,\otimes,\rightarrow_{\mathcal{P}})$ is equivalent to $(W,\otimes, \propto)$, where $W$ is the set of finite, nonempty multisets of nonzero complex numbers, and also to $(M_G,+,\propto)$, where $M_G$ is the set of finite, nonempty multisets over the (additive) group $G=\mathbb{R}\times \mathbb{R/Z}$. We can extend $S$ by allowing 0 entries in the vectors, but we will not go through the details.
    
    In particular, this shows that we should expect this transformation theory to behave similarly to the transformation theory of LU. Furthermore, if we were to also allow diagonal phase unitaries alongside permutation matrices, the resulting transformation theory would be equivalent to $\mathrm{EQ}_0$, and thus also to the TT of LU.
\end{rmk}

\section{Flexible catalysis of multisets}
\label{Flexible catalysis of multisets}

We saw in Section \ref{Transformation theories} that the TTs of LU and PM are equivalent to TTs of multisets over certain groups. So it makes sense to look at flexible catalysis in TTs of multisets separately. In Section \ref{Transformation theory of multisets over a group}, we give the necessary definitions. In Section \ref{Torsion-free groups}, we examine the case when the group is torsion-free. It turns out that the TTs arising from such groups behave in a very special way because torsion-free groups are orderable (see Propostion \ref{orderability}), and understanding them will help us when we turn to more general groups as well. Section \ref{General abelian groups} deals with general abelian groups, and it culminates in the proof of our main negative result about flexible catalysis of multisets. In Section \ref{Adding translation as a free operation}, we address a technical detail: we need to slightly modify our statements so that they also apply to the transformation theories of main interest, in which two multisets are equivalent if they are related by a translation. It is here that we reintroduce the TTs exactly as they already appeared in Section \ref{The transformation theories of LOCC, LU and PM}. Finally, in Section \ref{Advantages from flexibility}, we give evidence that flexible catalysis is stronger than traditional catalysis in some TTs of multisets, laying the foundations of the main positive results of this paper.

\subsection{The transformation theory of multisets over a group}
\label{Transformation theory of multisets over a group}

\begin{dfn} ($G$-multisets.)
    For an abelian group $(G, +)$, a $G$-multiset is a nonempty multiset with elements from $G$. We denote the collection of all finite $G$-multisets by $M_G$.
\end{dfn}

\begin{dfn}
\label{defn of addition of multisets}
    (Addition of multisets.)

    Let $(G, +)$ be an abelian group, and let $A,B$ be finite $G$-multisets. The sum $A+B$ is defined to be the multiset of all sums $a+b, a\in A, b\in B$, where the multiplicity of an element $x\in A+B$ is the number of ways it can be written as $x=a+b, a\in A, b\in B$. More precisely, if $m_S(e)$ denotes the multiplicity of $e$ in $S$, the sum of two multisets is defined so that
    \begin{equation}
        m_{A+B}(x)=\sum_{a+b=x} m_A(a)m_B(b).
    \end{equation}
\end{dfn}

\begin{prop}
    $T_G:=(M_G, +, =)$ is a symmetric transformation theory.
\end{prop}
\begin{proof}
    Straightforward.
\end{proof}

\begin{rmk}
    In this section we will mostly be working with transformation theories of the form $T_G=(M_G, +, =)$, which can be specified by fixing the group $G$. So we will simplify notation and write $\mathrm{Cat}_G$ instead of $\mathrm{Cat}_{T_G}$, and so on. Note also that the TTs of the form $(M_G, +, =)$ are not exactly what we are interested in. In the TTs of multisets equivalent to the TTs of LU, and PM, two multisets should be equivalent if they are equal up to translation (see Section \ref{The transformation theories of LOCC, LU and PM}). However, we start our analysis with these TTs because they are a bit easier to understand.
\end{rmk}

  \begin{prop} 
\label{extgroup}
  
      Let $H$ be a subgroup of the abelian group $G$. Then 
      
      (i) $\mathrm{Cat}_G\cap (M_H\times M_H)=\mathrm{Cat}_H,$

      (ii) $\mathrm{Ext}_G\cap (M_H\times M_H)=\mathrm{Ext}_H.$
  \end{prop}

\begin{rmk}
    This result says that allowing a catalyst or a discarded multiset to come from a larger group $G$ will not enable more transformations of multisets over $H$. Theorem \ref{catgdonthelp} states a similar result for specific groups $G,H$.
\end{rmk}

  \begin{proof}
      (i) Clearly $\mathrm{Cat}_G\cap (M_H\times M_H) \supseteq \mathrm{Cat}_H$, so we only need to show $\mathrm{Cat}_G\cap (M_H\times M_H) \subseteq \mathrm{Cat}_H$. Take $(A,B)\in \mathrm{Cat}_G\cap (M_H\times M_H)$, and assume that $C\in M_G$ is such that $A+C=B+C$.

    Notice that if $a+c_1=b+c_2$ for some  $a\in A, b\in B, c_1,c_2\in C$, then $c_1-c_2=b-a\in H$. Let us draw a graph with vertex set $C$ (with multiplicities), and connect two vertices corresponding to $c_1,c_2\in C$ iff there exist $a\in A,b\in B$ such that $a+c_1=b+c_2$ or $a+c_2=b+c_1$. Let $C'$ be the connected component of some $c\in C$. Then for any $x\in C'$, $x-c\in H$, i.e. $C'-c$ has elements from $H$. But we also have $A+C'=B+C'$, since $A+C=B+C$ and no vertex in $C\backslash C'$ is connected to any vertex in $C'$. Hence $A+(C'-c)=B+(C'-c)$, and thus $(A,B)\in \mathrm{Cat}_H$.

    (ii) Again, one inclusion is clear. To see that $\mathrm{Ext}_G\cap (M_H\times M_H)\subseteq \mathrm{Ext}_H$, take $A,B\in M_H$ and $D\in M_G$ such that $A=B+D$. Then any element of $D$ can be written as a difference of an element of $A$ and an element of $B$, and is thus an element of $H$.
      
  \end{proof}
  
\subsection{Torsion-free groups}
\label{Torsion-free groups}

It turns out that the torsion subgroup plays an important role in the catalytic structure of transformation theories of the form $T_G=(M_G, +, =)$. Finitely generated torsion-free groups have a simple structure, and they admit a translation-invariant total order (see Proposition \ref{orderability}), which makes it easier to analyse catalytic phenomena. They are also relevant because the transformation theory of LU is equivalent to $T_\mathbb{R}':=(M_\mathbb{R}, +, \propto)$ by Proposition \ref{chain_of_equivalences}, and $\mathbb{R}$ is torsion-free. Therefore, we start our analysis by assuming that $G$ is torsion-free, and only then will we move on to abelian groups in general.

\begin{prop}
\label{orderability}
    Any finitely generated torsion-free abelian group is (bi-) orderable. That is, it admits a translation-invariant total order.
\end{prop}
\begin{rmk}
    In fact, a much stronger result is known: an abelian group is bi-orderable if and only if it is torsion-free \cite{Levi1942}. However, the statement of the proposition will be sufficient for our purposes.
\end{rmk}
\begin{proof}
    By the structure theorem for finitely generated abelian groups, there is some $n\in \mathbb{N}$ such that our group is isomorphic to $\mathbb{Z}^n$. But the lexicographical order is a translation-invariant total order on $\mathbb{Z}^n$.
\end{proof}

\begin{prop}
    \label{cancellation}
    Let $A,B,C$ be finite $G$-multisets for a torsion-free abelian group $G$. If $A+C=B+C$, then $A=B$. In other words, $\mathrm{Cat}_G=\mathrm{Tr}_G$.
\end{prop}
\begin{proof}
    First observe that by finiteness of $A,B,C$, we may assume that $G$ is finitely generated. Thus, by Proposition \ref{orderability}, $G$ is orderable. Also, $A$ and $B$ must have the same size, say $d$. We will use induction on $d$. By looking at the maximal element of $A+C=B+C$, it is clear that the maximal elements of $A$ and $B$ must agree. Therefore, we can remove this element from $A$ and $B$, and then apply the induction hypothesis for the remaining sets with the same $C$ as before.
\end{proof}
\begin{prop} 
\label{halving}
    Let $A,B$ be finite $G$-multisets for a torsion-free abelian group $G$. If $A+A=B+B$, then $A=B$.
\end{prop}
\begin{proof}
        We may assume by finiteness of $A,B$ that $G$ is finitely generated. Order $G$ in accordance with Proposition \ref{orderability}. Let $c$ and $d$ be the ordered lists (``vectors") of elements of $A$ and $B$, respectively, such that the coordinates of $c,d$ are in decreasing order. Clearly $c$ and $d$ must have the same dimension $n$.  Then we need to show that $c=d$. Assume for a contradiction that $c\neq d$, and let $k$ be minimal such that $c_k\neq d_k$. Assume without loss of generality that $c_k>d_k$. Then
        $$\{(i,j)\in [n]^2: c_i+c_j \geq c_1+c_k \} \supset \{(i,j)\in [k-1]^2: c_i+c_j \geq c_1+c_k \} \cup \{(1,k)\}$$
        and
           $$\{(i,j)\in [n]^2: d_i+d_j \geq c_1+c_k \} = \{(i,j)\in [k-1]^2: d_i+d_j \geq c_1+c_k \}$$
           since the largest $d_i+d_j$ subject to $(i,j)\notin [k-1]^2$ is $d_1+d_k<c_1+c_k$.
           But by minimality of $k$, the sets $\{(i,j)\in [k-1]^2: c_i+c_j \geq c_1+c_k \} $ and $\{(i,j)\in [k-1]^2: d_i+d_j \geq c_1+c_k \}$ are equal, and hence 
            $$\{(i,j)\in [n]^2: c_i+c_j \geq c_1+c_k \}\supsetneq \{(i,j)\in [n]^2: d_i+d_j \geq c_1+c_k \},$$
            a contradiction.
    \end{proof}

    \begin{prop} 
\label{divbyn}
    Let $A,B$ be finite $G$-multisets for a torsion-free abelian group $G$. If $nA=nB$ (where $nX$ denotes $X+X+\ldots +X$ with $n$ terms in the sum), then $A=B$.
\end{prop}
\begin{proof}
    This is a straightforward generalization of Proposition \ref{halving}.
\end{proof}

 \begin{rmk}
        \cite{nathanson1978representation} proves a similar result for sets of natural numbers.
    \end{rmk}

\begin{lem}
\label{flexible_multiset_catalysis_is_useless_in_a_torsion_free_group}
    Let $G$ be a torsion-free abelian group. Then $(A,B)\in \mathrm{Cat}_G^{(\mathrm{fin})}\implies A=B$.

\end{lem}
\begin{proof}
By Proposition \ref{Union_of_cyclic_classes}, $(A,B)\in \mathrm{Cat}_G^{(n)}$ for some $n\in \mathbb{Z}_{>0}$. Therefore, by Proposition, \ref{flexvsmulticop} $(nA,nB)\in \mathrm{Cat}_G$. Then, by Proposition \ref{cancellation}, we have $nA=nB$. Finally, by Proposition \ref{divbyn}, we have $A=B$.
\end{proof}

\begin{eg}
    Take the example $G=\mathbb{R}$. Then Lemma \ref{flexible_multiset_catalysis_is_useless_in_a_torsion_free_group} says that flexible catalysis with a finite catalytic set does not enable any new transformations in $T_\mathbb{R}$. And by Corollary \ref{ttlueq}, $T_\mathbb{R}$ is almost equivalent to the TT of LU.
\end{eg}

\subsection{General abelian groups}
\label{General abelian groups}

\begin{prop}
\label{fincatprojeq}
    Let $G$ be a finitely generated abelian group. Assume that $G=T\times F$, where $T$ is arbitrary and $F$ is torsion-free. Let $\pi: G\rightarrow F$ be the projection map. Then $(A,B)\in \mathrm{Cat}_G^{(\mathrm{fin})}\implies \pi (A)=\pi (B)$.
\end{prop}
\begin{proof}
    Since $\pi$ is a group homomorphism, it sends the defining equations of $\mathrm{Cat}_G^{(\mathrm{fin})}$ to those of $\mathrm{Cat}_F^{(\mathrm{fin})}$. Therefore, $(A,B)\in \mathrm{Cat}_G^{(\mathrm{fin})}\implies (\pi(A),\pi(B)) \in \mathrm{Cat}_F^{(\mathrm{fin})}$. Then we get $\pi(A)=\pi(B)$ by Lemma \ref{flexible_multiset_catalysis_is_useless_in_a_torsion_free_group}.
\end{proof}

\begin{prop}
\label{projections_equal_implies_single_cat}
    Let $G=T\times F$ be a finitely generated abelian group, where $T$ is finite. Let $\pi: G\rightarrow F$ be the projection map. Then $\pi (A)=\pi (B) \implies (A,B)\in \mathrm{Cat}_G(T)$.
\end{prop}

\begin{proof}
    With a slight abuse of notation we can treat $T$ as a subgroup of $G$, and hence also as a $G$-multiset. Then we have $A+T=\pi(A)+T=\pi(B)+T=B+T$. (Where we used the fact that for any $x\in G$, $x+T=\pi(x)+T$.)
\end{proof}

\begin{thm}
    For any abelian group $G$, we have $\mathrm{Cat}_G^{(\mathrm{fin})}=\mathrm{Cat}_G$.
\end{thm}
\begin{proof}
    Only the $\mathrm{Cat}_G^{(\mathrm{fin})}\subseteq \mathrm{Cat}_G$ direction is nontrivial. Let $(A,B)\in \mathrm{Cat}_G(S)$ for a finite $S\in S_G$. Since $S$ is finite, and each element of $S$ is a finite $G$-multiset, we may assume that $G$ is finitely generated. Write $G=T\times F$, where $T$ is the torsion group of $G$ and $F$ is torsion-free. Let $\pi: G\rightarrow F$ be the projection map. Then $(A,B)\in \mathrm{Cat}_G^{(\mathrm{fin})}\implies \pi (A)=\pi (B)$, by Proposition \ref{fincatprojeq}. But by Proposition \ref{projections_equal_implies_single_cat}, this means that $(A,B)\in \mathrm{Cat}_G(T)\subseteq \mathrm{Cat}_G$.
\end{proof}

\subsection{Adding translation as a free operation}
\label{Adding translation as a free operation}

In this section, we will extend our previous results to transformation theories of the form $T_G'=(M_G, +, \propto)$, where $\propto$ denotes equality up to translation by a group element. We will only work with transformation theories of this form, so we will refer to the transformation theory by naming the underlying group $G$.

\begin{prop}
    \label{cancpv}
    Let $A,B,C$ be finite $G$-multisets for a torsion-free abelian group $G$, and $g\in G$. If $A+C+g=B+C$, then $A+g=B$. In particular, $(A,B)\in \mathrm{Cat}_G \implies A\propto B$.
\end{prop}
\begin{proof}
    Apply Proposition \ref{cancellation} to the pair $(A+g,B)$.
\end{proof}
\begin{prop} 

\label{divbynpv}
    Let $A,B$ be finite $G$-multisets for a torsion-free abelian group $G$. If $nA\propto nB$, then $A\propto B$.
\end{prop}
\begin{proof}
    We may assume that $G$ is finitely generated, and let $g\in G$ be such that $nA+g=nB$. Look at the minimal elements of $nA+g$ and $nB$ to deduce that $g=nh$ for some $h\in G$. Thus, $n(A+h)=nB$, and hence $A+h=B$ by Proposition \ref{divbyn}.
\end{proof}

\begin{lem}
\label{flexmulti}
    Let $G$ be a torsion-free abelian group. Then $(A,B)\in \mathrm{Cat}_G^{(\mathrm{fin})}\implies A\propto B$.

\end{lem}
\begin{proof}
By Proposition \ref{Union_of_cyclic_classes} and Proposition \ref{flexvsmulticop}, $(nA,nB)\in \mathrm{Cat}_G$ for some $n$. Then, by Proposition \ref{cancpv}, we have $nA\propto nB$. But then, by Proposition \ref{divbynpv}, we have $A\propto B$.
\end{proof}

\begin{prop}
\label{fincateqpv}
    Let $G$ be a finitely generated abelian group. Assume that $G=T\times F$, where $T$ is arbitrary and $F$ is torsion-free. Let $\pi: G\rightarrow F$ be the projection map. Then $(A,B)\in \mathrm{Cat}_G^{(\mathrm{fin})}\implies \pi (A)\propto\pi (B)$.
\end{prop}
\begin{proof}
    Since $\pi$ is a group homomorphism, it sends the defining equations of $\mathrm{Cat}_G^{(\mathrm{fin})}$ to those of $\mathrm{Cat}_F^{(\mathrm{fin})}$. Therefore, $(A,B)\in \mathrm{Cat}_G^{(\mathrm{fin})}\implies (\pi(A),\pi(B)) \in \mathrm{Cat}_F^{(\mathrm{fin})}$. Then we get $\pi(A)\propto\pi(B)$ by Lemma \ref{flexmulti}.
\end{proof}

\begin{prop}
\label{projimp}
    Let $G=T\times F$ be a finitely generated abelian group, where $T$ is finite. Let $\pi: G\rightarrow F$ be the projection map. Then $\pi (A)\propto \pi (B) \implies (A,B)\in \mathrm{Cat}_G(T)$.
\end{prop}

\begin{proof}
Let $g\in F$ be such that $\pi(A)+g=\pi(B)$. Then
    $$A+T+g=\pi(A)+T+g=\pi(A)+g+T=\pi(B)+T=B+T.$$
\end{proof}

\begin{thm}
\label{abelianeq}
    For any abelian group $G$, we have $\mathrm{Cat}_G^{(\mathrm{fin})}=\mathrm{Cat}_G$.
\end{thm}
\begin{proof}
    Only the $\mathrm{Cat}_G^{(\mathrm{fin})}\subseteq \mathrm{Cat}_G$ direction is nontrivial. Let $(A,B)\in \mathrm{Cat}_G(S)$ for a finite $S\in S_G$. Since $S$ is finite, and each element of $S$ is a finite $G$-multiset, we may assume that $G$ is finitely generated. Write $G=T\times F$, where $T$ is the torsion group of $G$ and $F$ is torsion-free. Let $\pi: G\rightarrow F$ be the projection map. Then $(A,B)\in \mathrm{Cat}_G^{(\mathrm{fin})}\implies \pi (A)\propto \pi (B)$, by Proposition \ref{fincateqpv}. But by Proposition \ref{projimp}, this means that $(A,B)\in \mathrm{Cat}_G(T)\subseteq \mathrm{Cat}_G$.
\end{proof}

\begin{rmk}
    The proof of Theorem \ref{abelianeq} also implies $\mathrm{Cat}_G^{(\mathrm{fin})}=\{(A,B) \in M_G\times M_G: \pi(A)\propto \pi(B)\}=\mathrm{Cat}_G$, where $\pi$ is the projection onto the torsion-free factor of $G$. This shows that the catalytic strucure of transformation theories of the form $T_G'=(M_G, +, \propto)$ depends on the torsion subgroup of $G$. In particular, there exist nontrivial catalytic transformations if and only if the torsion subgroup is nontrivial.
\end{rmk}

\subsection{Advantages from flexibility}
\label{Advantages from flexibility}

In this section, we present results showing that there is something to be gained from allowing flexibility in catalytic transformations of multisets. Our previous results, particularly Theorem \ref{abelianeq}, imply restrictions on the potential benefits of flexibility for catalytic transformations. We will see that despite these negative results, it is possible to identify a weaker form of advantage. However, the most significant benefits will only appear when we look at extractions. Therefore, we will be working with extractions for most of this section. Notice that for extractions it does not make a difference whether we are working in the transformation theory $T_G$ or $T_G'$, because the translation constant can always be absorbed by the discarded multiset, so we can refer to catalytic extraction classes by specifying the type of catalysis we wish to allow and naming the underlying group.

\begin{eg}
    Let $\omega = e^{2\pi i /3}$ and let $A=[1,\omega, \omega]/\sqrt{3}$, $B=[1,\omega^2, \omega^2]/\sqrt{3}$. Then $A\otimes A$ and $B\otimes B$ are permutation-equivalent, and clearly $A\otimes B$ and $B\otimes A$ are permutation equivalent. But $A\otimes A$ and $B\otimes A$ are not, not even up to global phase, nor are $A\otimes B$ and $B\otimes B$.

   In other words, for $G=\mathbb{C}^*$ and $S=\{A,B\}$ (where now $A,B$ denote the corresponding multisets), we have $(A,B)\in \mathrm{Cat}_G^{(f)}(S)$, but for no element $C\in S$ do we have $(A,B)\in \mathrm{Cat}_G(C)$.

    This is already some evidence that it makes sense to think about flexibility, since if we have access to either element of the catalyst set $S$, we can perform an arbitrary number of transformations $A\rightarrow B$, but only if we allow the catalyst state to alternate (in this case between $A$ and $B$).
\end{eg}

\begin{dfn}
\label{polys to multisets}
    Let $\iota:\mathbb{Z}_{\geq 0}[x]\backslash\{0\}\rightarrow M_\mathbb{Z}$ be the map that sends any polynomial $a_0+a_1x+\ldots a_nx^n$ with nonnegative integer coefficients to the multiset whose elements come from $\{0,1,\ldots,n\}$ and the multiplicity of any $k\in \{0,1,\ldots,n\}$ is the coefficient $a_k$ of the degree $k$ term of the polynomial. Write this multiset as $\{0^{(a_0)},1^{(a_1)},\ldots, n^{(a_n)}\}$.
\end{dfn}

\begin{prop}
\label{polys_to_multisets}
    The map $\iota:\mathbb{Z}_{\geq 0}[x]\backslash\{0\}\rightarrow M_\mathbb{Z}, a_0+a_1x+\ldots a_nx^n \mapsto \{0^{(a_0)},1^{(a_1)},\ldots, n^{(a_n)}\}$ is an injective monoid homomorphism, and its image is the set of finite $\mathbb{Z}$-multisets with nonnegative elements.
\end{prop}
\begin{proof}
    The only non-trivial claim to prove is that $\iota(pq)=\iota(p)+\iota(q)$ for all $p,q\in \mathbb{Z}_{\geq 0}[x]$. To see why it's true, notice that the coefficient of $x^k$ in
    \begin{equation}
        p(x)q(x)=(a_0+a_1x+\ldots a_nx^n )(b_0+b_1x+\ldots b_nx^n )
    \end{equation}
    is
    \begin{equation}
        \sum_{j=0}^k a_jb_{k-j},
    \end{equation}
    and the multiplicity of $k$ in $\iota(p)+\iota(q)$ is the same number.
\end{proof}

\begin{eg}
\label{advantage}
    Here we will present an element $(A,B)$ of $\mathrm{CatExt}_\mathbb{Z}^{(\mathrm{fin})}\backslash \mathrm{CatExt}_\mathbb{Z}$ for which there exists a flexible catalytic procedure using two catalysts of size smaller than $A$.
    
      Define the polynomials $X(x)=4+x,Y(x)=2+2x-x^2+2x^3+2x^4, B(x)=1+x, A(x)=B(x)X(x)Y(x), D_1=X^2, D_0=Y^2,C_0=BX, C_1=BY$. We claim that, the multisets corresponding to the polynomials under the map $\iota$ defined in Proposition \ref{polys_to_multisets} satisfy
     \begin{equation}
         \iota(A)+\iota(C_0)=\iota(B)+\iota(D_1)+\iota(C_1)
     \end{equation}
     and
\begin{equation}
         \iota(A)+\iota(C_1)=\iota(B)+\iota(D_0)+\iota(C_0),
     \end{equation}
     thus $(\iota(A),\iota(B))\in \mathrm{CatExt}_\mathbb{Z}^{(f)}(\{\iota(C_0),\iota(C_1)\})$. Moreover, we claim that $(\iota(A),\iota(B))\notin \mathrm{CatExt}_\mathbb{Z}$.

     Firstly, it must be checked that $A,B,C_i,D_i$ all have nonnegative coefficients, but this is straightforward.

     Secondly, notice that the equations above follow immediately from Proposition \ref{polys_to_multisets} and the polynomial identities
     \begin{equation}
         A(x)C_0(x)=B(x)D_1(x)C_1(x)
     \end{equation}
     and
      \begin{equation}
         A(x)C_1(x)=B(x)D_0(x)C_0(x).
     \end{equation}

     Finally, assume for a contradiction that $\iota(A)+U=\iota(B)+V+U$ for some $U,V\in M_\mathbb{Z}$. We may assume without loss of generality that $0$ is the minimal element of $U$. Then, since $0$ is also the minimal element of $\iota(A)$ and $\iota(B)$, we must also have that $0$ is the minimal element of $V$. In particular, $U,V$ lie in the image of $\iota$, i.e. there exist polynomials $C,D\in\mathbb{Z}_{\geq 0}[x]$ such that
     \begin{equation}
         \iota(A)+\iota(C)=\iota(B)+\iota(D)+\iota(C).
     \end{equation}
     Therefore, by Proposition \ref{polys_to_multisets},
      \begin{equation}
         \iota(AC)=\iota(BDC),
     \end{equation}
     and by injectivity this gives $A(x)C(x)=B(x)D(x)C(x)$. Note that $C$ is not the constant zero polynomial, so we can divide by it to obtain $A(x)=B(x)D(x)$. But we also have $A(x)=B(x)X(x)Y(x)$, so (by unique factorisation) we get $X(x)Y(x)=D(x).$ But $X(x)Y(x)$ has a negative coefficient, giving the contradiction.

    For concreteness, we provide the description of the multisets in terms of their elements:
     $$\iota(A)=\{0^{(8)}, 1^{(18)},2^{(8)},3^{(5)}, 4^{(17)}, 5^{(12)}, 6^{(2)}\},\iota(B)=\{0,1\}, \iota(C_0)=\{0^{(4)}, 1^{(5)},2^{(1)}\},$$
     $$\iota(C_1)=\{0^{(2)}, 1^{(4)},2^{(1)},3^{(1)}, 4^{(4)}, 5^{(2)}\}, \iota(D_0)=\{0^{(4)}, 1^{(8)},3^{(4)}, 4^{(17)}, 5^{(4)}, 7^{(8)}, 8^{(4)}\}, \iota(D_1)=\{0^{(16)}, 1^{(8)},2^{(1)}\}.$$
     So $A$ corresponds to a Schmidt vector of dimension 70, whereas $C_0$ corresponds to a 10-dimensional Schmidt vector. So this example not only shows that flexible catalysis can be stronger than traditional catalysis in an absolute sense, but also that it is more efficient in terms of space requirements than any catalytic multicopy transformation simulating it.
\end{eg}

\begin{dfn}
(Negativity of a polynomial.)
    For a polynomial $p\in\mathbb{Z}[x_1,\ldots,x_m]=\mathbb{Z}[\vec{x}]$, the negativity of $p$ is defined to be the smallest positive integer $n$ such that $p^n$ has nonnegative coefficients: $n(p):=\min \{n\in \mathbb{Z}_{>0}: p^n\in \mathbb{Z}_{\geq 0}[\vec{x}]\}$. If there is no such positive integer, say that $p$ is infinitely negative and write $n(p)=\infty$.
\end{dfn}

\begin{prop}

\label{anyposint}
    For any positive integer $n$, there exists a polynomial $p\in \mathbb{Z}[x]$ such that $n=n(p)$. Moreover, $p$ can be chosen such that $(1+x)p(x)\in  \mathbb{Z}_{\geq 0}[x]$.
\end{prop}

\begin{proof}
    The statement is clear for $n=1$, so let us assume that $n\geq 2$. We will construct a polynomial $p(x)=a_0+a_1x+a_2x^2+a_3x^3+a_4x^4$ such that $n(p)=n$.

    Firstly, notice that the coefficient of $x^2$ in $p(x)^k$ is 
    \begin{equation}
        ka_0^{k-1}a_2+\binom{k}{2}a_1^2a_0^{k-2}=ka_0^{k-2}\left( a_0a_2+\frac{k-1}{2}a_1^2\right).
    \end{equation}

    Therefore, choosing $a_0,a_1,a_2$ in such a way that 
    \begin{equation}
    \label{coeffbound}
        \frac{n-2}{2}a_1^2 < - a_0a_2\leq \frac{n-1}{2}a_1^2
    \end{equation}
    will guarantee that the coefficient of $x^2$ in $p(x)^k$ is negative precisely when $k<n$. So all we need to do now is make sure that all the coefficients of $p(x)^n$ are nonnegative.

    Let $a_2=-1, a_1=5^n, a_0=a_3=a_4=\lfloor \frac{n-1}{2}5^{2n}\rfloor.$ Then the constant and linear coefficients of $p(x)^n$ are clearly positive, and the quadratic coefficient is also nonnegative since this choice satisfies equation (\ref{coeffbound}).

    If $4n\geq j\geq 3$, then $j$ can be written as a sum of $0$'s, $3$'s and $4$'s and at most a single $1$, where the total number of terms is $n$. Therefore, there is at least one degree $j$ term with a positive coefficient of size at least $a_0^{n-1}a_1$ in the expansion of $p(x)^n$. At the same time, $5^n\times  a_0^{n-1}a_2$ is a bound on the maximal amount of total negative contribution towards the coefficient of $x^j$ that we can get from different negative terms of the expansion of $p(x)^n$: there are $5^n$ terms in the expansion, and any negative term must have a factor of $a_2$, so the absolute value of each negative term is at most $a_0^{n-1}|a_2|=a_0^{n-1}$.
    
    It follows that the coefficient of $x^j$ is at least
   
    \begin{equation}
        a_0^{n-1}a_1+5^n\times  a_0^{n-1}a_2=0.
    \end{equation}

    Finally, we note that this construction clearly satisfies $(1+x)p(x)\in  \mathbb{Z}_{\geq 0}[x]$.
\end{proof}

\begin{prop}
\label{arbnumreq}
    For any $n\in \mathbb{Z}_{> 1}$, $\mathrm{CatExt}_\mathbb{Z}^{(n)}\backslash \mathrm{CatExt}_\mathbb{Z}^{(\leq n-1)}\neq \emptyset$.
\end{prop}

\begin{proof}
    By Proposition \ref{anyposint}, there is a polynomial $D\in \mathbb{Z}[x]$ such that $n(D)=n$, and $A,B\in \mathbb{Z}_{\geq 0}[x]$ such that $A=BD$ and $B(0)>0$. (For example, $B(x)=1+x$.) Let $\iota$ be the map from Definition \ref{polys to multisets}. Then $A^n=B^nD^n$, so $n\iota(A)=n\iota(B)+\iota(D^n)$, and thus, by Proposition \ref{flexvsmulticopyextrtt}, $(\iota(A),\iota(B))\in \mathrm{CatExt}_\mathbb{Z}^{(n)}$.

    Now assume for a contradiction that $(\iota(A),\iota(B))\in \mathrm{CatExt}_\mathbb{Z}^{(k)}$ for some $k<n$. Then, by Proposition \ref{flexvsmulticopyextrtt}, there exist $U,V\in M_\mathbb{Z}$ such that $k\iota(A)+U=k\iota(B)+V+U$. But $\mathbb{Z}$ is torsion-free, so by Proposition \ref{cancellation} we must also have $k\iota(A)=k\iota(B)+V$. Hence all the elements of $V$ are non-negative, and thus $V=\iota(P)$ for some $P\in \mathbb{Z}_{\geq 0}[x]$. But then $\iota(A^k)=\iota(B^kP)$, and by injectivity this yields $A^k=B^kP$. Thus, by unique factorisation, we must have $P=D^k$, contradicting the fact that $D^k$ has a negative coefficient.
\end{proof}

\begin{thm}
\label{catgdonthelp}
    Let $G=\mathbb{(R,+)}\times\mathbb{(R/Z,+)}$. For any $n\in \mathbb{Z}_{> 0}$, 
    \begin{equation}
   \mathrm{CatExt}_G^{(\leq n) }\cap (M_\mathbb{Z}\times M_\mathbb{Z})= \mathrm{CatExt}_\mathbb{Z}^{(\leq n)}.    
    \end{equation}
\end{thm}

\begin{rmk}
    Here $\mathbb{Z}$ is considered a subgroup of $G$ via the isomorphism $\mathbb{Z}\cong \mathbb{Z}\times \{0+\mathbb{Z}\}\leq \mathbb{R}\times\mathbb{R/Z}=G$. So the statement means that we do not get more flexible catalytic extractions of $\mathbb{Z}$-multisets by allowing the elements of the catalysts and discarded multisets to come from the larger group $G$. Note also that this statement is essentially a more complicated version of Proposition \ref{extgroup}, with the choices $G=\mathbb{(R,+)}\times\mathbb{(R/Z,+)}$ and $H=\mathbb{Z}$.
\end{rmk}

\begin{proof}
    The inclusion $\supseteq$ is clear. For $\subseteq$, assume that
    
    \begin{equation}
    (A,B)\in \mathrm{CatExt}_G^{(\leq n) }\cap (M_\mathbb{Z}\times M_\mathbb{Z}). 
    \end{equation}
    
    Then, by Proposition \ref{Union_of_cyclic_classes} and Proposition \ref{flexvsmulticopyextrtt}, $(kA,kB)\in \mathrm{CatExt}_G$ for some $1\leq k\leq n$, i.e. $kA+C=kB+D+C$ for some $C,D\in M_G$. Let $\pi:G\rightarrow\mathbb{R}$ be the projection map. Then  $k\pi(A)+\pi(C)=k\pi(B)+\pi(D)+\pi(C)$. But $A,B\in M_\mathbb{Z}\subset M_\mathbb{R}$, so $\pi(A)=A$ and $\pi(B)=B$. Therefore, $(kA, kB+\pi(D))\in \mathrm{Cat}_\mathbb{R}$. But $\mathbb{R}$ is torsion-free, so by Lemma \ref{flexible_multiset_catalysis_is_useless_in_a_torsion_free_group}, we must have $kA=kB+\pi(D)$, and thus $\pi(D)\in M_\mathbb{Z}$. Therefore, by applying  Proposition \ref{flexvsmulticopyextrtt} again, we conclude that $(A,B)\in \mathrm{CatExt}_\mathbb{Z}^{(k)}\subseteq\mathrm{CatExt}_\mathbb{Z}^{(\leq n)}$.
    
\end{proof}

\begin{cor} 
\label{anynumcat}
    Let $G=\mathbb{(R,+)}\times\mathbb{(R/Z,+)}$. For any $n\in \mathbb{Z}_{> 1}$,
    \begin{equation}
        \mathrm{CatExt}_\mathbb{Z}^{(n)}\backslash \mathrm{CatExt}_G^{(\leq n-1)}\neq \emptyset.
    \end{equation}
\end{cor}
\begin{proof}
    By Proposition \ref{arbnumreq}, $\mathrm{CatExt}_\mathbb{Z}^{(n)}\backslash \mathrm{CatExt}_\mathbb{Z}^{(\leq n-1)}\neq \emptyset$. But Theorem \ref{catgdonthelp} implies that
    \begin{equation}
        \mathrm{CatExt}_\mathbb{Z}^{(n)}\backslash \mathrm{CatExt}_\mathbb{Z}^{(\leq n-1)}= \mathrm{CatExt}_\mathbb{Z}^{(n)}\backslash \mathrm{CatExt}_G^{(\leq n-1)}.
    \end{equation}
\end{proof}

\section{Flexible catalysis in quantum information}
\label{Flexible catalysis in quantum information}

In this section, we give the formal statements and the proofs of our main results. Most of these are just corollaries of our previous results, and their proofs follow the same pattern: we use a result from Section \ref{Transformation theories} to relate the TT of interest to a TT of multisets, and quote the corresponding result from Section \ref{Flexible catalysis of multisets}. We will only need new ideas for LOCC.

\begin{thm}
\label{CatExtLUfin_beats_CatExtLU}
    \begin{equation}
        \mathrm{CatExt}_{\mathrm{LU}}\subsetneq \mathrm{CatExt}_{\mathrm{LU}}^{(\mathrm{fin})}.
    \end{equation}
    Moreover, for any $n\in \mathbb{Z}_{> 1}$, 
    \begin{equation}
        \mathrm{CatExt}_{\mathrm{LU}}^{(n)}\backslash \mathrm{CatExt}_{\mathrm{LU}}^{(\leq n-1)}\neq \emptyset.
    \end{equation}
\end{thm}
\begin{proof}
   Corollary \ref{anynumcat} implies that 
    \begin{equation}
        \mathrm{CatExt}_{T_\mathbb{R}}^{(n)}\backslash \mathrm{CatExt}_{T_\mathbb{R}}^{(\leq n-1)}\neq \emptyset.
    \end{equation}
    Note that $\mathrm{CatExt}_{T_\mathbb{R}'}^{(k)}=\mathrm{CatExt}_{T_\mathbb{R}}^{(k)}$ for any $k$, since we can always let the discarded multiset absorb the translation constant. So we have
    \begin{equation}
        \mathrm{CatExt}_{T_\mathbb{R}'}^{(n)}\backslash \mathrm{CatExt}_{T_\mathbb{R}'}^{(\leq n-1)}\neq \emptyset.
    \end{equation}
     By Corollary \ref{ttlueq}, the TT of $\mathrm{LU}$ is equivalent to $T_\mathbb{R}'=(M_\mathbb{R},+,\propto)$, concluding the proof.
\end{proof}

\begin{thm}
    \label{finfcluuseless}
    $\mathrm{Cat}_{\mathrm{LU}}^{(\mathrm{fin})}=\mathrm{Tr}_{\mathrm{LU}}$.
\end{thm}
\begin{proof}

    Lemma \ref{flexmulti} states that finite flexible catalysis cannot enhance transformations in TTs of the form $T_G'$ if $G$ is torsion-free. By Corollary \ref{ttlueq}, the TT of LU is equivalent to $T_\mathbb{R}'$, and $(\mathbb{R},+)$ is a torsion-free group. It follows that $\mathrm{Cat}_{\mathrm{LU}}^{(\mathrm{fin})}=\mathrm{Tr}_{\mathrm{LU}}$.
\end{proof}

\begin{rmk}
    It is interesting to note that Theorem \ref{finfcluuseless} implies that (traditional) catalysis adds no power to LU extraction procedures. But, by Theorem \ref{CatExtLUfin_beats_CatExtLU}, flexible catalytic extraction is stronger, and any number of catalysts may be needed. Moreover, by Proposition \ref{infeqtrcat}, we also have $\mathrm{Cat}_{\mathrm{LU}}^{(\mathrm{fin})}\subsetneq \mathrm{Cat}_{\mathrm{LU}}^{(f)}$.

\end{rmk}

\begin{center}
    
\includegraphics[width=11cm]{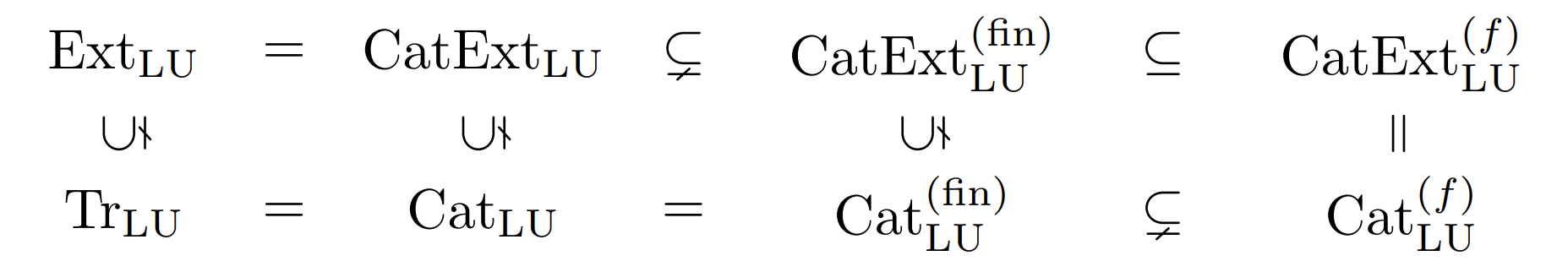}

\textbf{Figure 2}: Relations between the different types of catalytic LU transformation classes.
\end{center}

\vspace{10pt}

\begin{prop}
\label{LOCCpositive}
    There exists a set $S$ of two bipartite quantum states such that
    \begin{equation}
        \mathrm{Cat}_{\mathrm{LOCC}}^{(f)}(S)\backslash \bigcup_{C\in S} \mathrm{Cat}_{\mathrm{LOCC}}(C)\neq \emptyset.
    \end{equation}
\end{prop}

\begin{proof}
        Let $\ket{\psi_1}= \ket{\eta_1}$ with squared Schmidt coefficients given by the vector $a=(0.4,0.4,0.1,0.1)$ and $\ket{\psi_2}=\ket{\eta_2}$ with squared Schmidt coefficients given by the vector $b=(0.5,0.29,0.21,0)$. It can be checked that $a\otimes a\prec b\otimes b$, and clearly $a\otimes b\prec b\otimes a$, i.e. the transformations $\ket{\psi_1}\otimes\ket{\eta_1}\rightarrow\ket{\psi_2}\otimes\ket{\eta_2}$ and $\ket{\psi_1}\otimes\ket{\eta_2}\rightarrow\ket{\psi_2}\otimes\ket{\eta_1}$ are possible in LOCC. However, the transformations $\ket{\psi_1}\otimes\ket{\eta_1}\rightarrow\ket{\psi_2}\otimes\ket{\eta_1}$ and $\ket{\psi_1}\otimes\ket{\eta_2}\rightarrow\ket{\psi_2}\otimes\ket{\eta_2}$ are not possible in LOCC, because the corresponding majorization relations do not hold. So, for $S=\{\ket{\eta_1},\ket{\eta_2}\}$, we have
        
        \begin{equation}
            (\ket{\psi_1},\ket{\psi_2})\in \mathrm{Cat}_{\mathrm{LOCC}}^{(f)}(S)\backslash \bigcup_{C\in S} \mathrm{Cat}_{\mathrm{LOCC}}(C).
        \end{equation}
\end{proof}

\begin{thm}
\label{LOCC_negative_result}
    $\mathrm{Cat}_{\mathrm{LOCC}}^{(\mathrm{fin})}=\mathrm{Cat}_{\mathrm{LOCC}}$.
    
\end{thm}

\begin{proof}
 By Proposition \ref{Union_of_cyclic_classes}, $\mathrm{Cat}_{\mathrm{LOCC}}^{(\mathrm{fin})}=\cup_{n\in \N} \mathrm{Cat}_{\mathrm{LOCC}}^{(n)}$. Therefore, by Proposition \ref{flexvsmulticop}, $\mathrm{Cat}_{\mathrm{LOCC}}^{(\mathrm{fin})}$ is equivalent to the combination of MLOCC and ELOCC, as defined in  \cite{duan2005multiple}. (Where by equivalent we mean that they make the same transformations possible.) By Theorem 2 in \cite{duan2005multiple}, the combination of MLOCC and ELOCC is equivalent to ELOCC, which is $\mathrm{Cat}_{\mathrm{LOCC}}$ using our notation.
\end{proof}

\begin{rmk}
    Theorems \ref{LOCCpositive} and \ref{LOCC_negative_result} show that flexibility (with finitely many catalysts) can give an advantage for catalytic LOCC transformations, but only if there is a restriction on the set of catalysts we are allowed to use.
\end{rmk}

\begin{center}

\includegraphics[width=12cm]{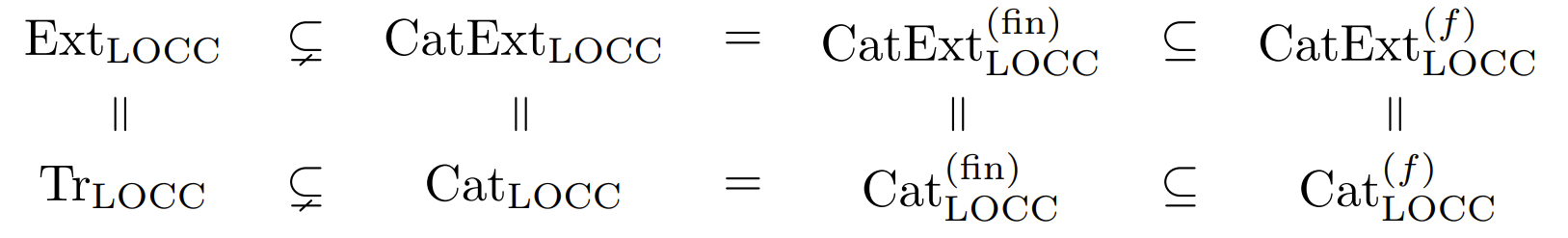}

\textbf{Figure 3}: Inclusion diagram of catalytic LOCC transformation classes. Note that LOCC contains discard operations, so extractions under LOCC are the same as transformations under LOCC.
    
\end{center}

\vspace{10pt}

Recall that according to Proposition \ref{description of TTs of LOCC and LU}, the TT of LOCC is equivalent to the TT $(V_0,\otimes, \prec)$, where $V_0$ is the set of discrete probability distribution vectors up to permutation. For $y \in V_0$, Nielsen defined $T(y)=\{x\in V_0: (x,y) \in \mathrm{Cat}_{V_0}\}$ \cite{nielsen2002majorization}. Similarly, we can define $M(y)=\{x\in V_0: \exists n \text{ s. t. } x^{\otimes n}\prec y^{\otimes n} \}$. So the fact that MLOCC is weaker than ELOCC \cite{duan2005multiple} can be written as $M(y)\subseteq T(y) \ \forall y\in V_0$. It was shown in \cite{CatalyticMajorizationAndell_pNorms} that $\overline{T(y)}=\overline{M(y)} \ \forall y$, and that $x\in \overline{T(y)} \iff ||x||_p\leq ||y||_p \ \forall p\geq 1$, where $||\cdot ||_p$ denotes the $\ell_p$-norm of a vector, and the closures are taken with respect to the $\ell_1$-norm. It is natural to extend these results to flexible catalysis as follows.

\begin{dfn}
    For $y\in V_0$, write $F(y)=\{x\in V_0: (x,y) \in \mathrm{Cat}^{(f)}_{V_0}\}$.
\end{dfn}

\begin{prop}
\label{LOCC closures negative result}
    $\overline{F(y)} = \overline{T(y)}.$ 
\end{prop}
\begin{proof}

Since we have $T(y)\subseteq F(y)$, it suffices to show that $F(y)\subseteq \overline{T(y)}.$ Take $x\in F(y)$. Then there exist probability vectors $c_0,c_1,\ldots$ such that $x\otimes c_i\prec y\otimes c_{i+1}$ for all $i=0,1,\ldots$. Since the $\ell_p$-norms are multiplicative Schur-convex functions, this implies
    \begin{equation}
    \label{products of ellp norms}
        ||x||_p\cdot ||c_i||_p\leq ||y||_p\cdot ||c_{i+1}||_p \ \ \forall p\geq 1.
    \end{equation}

    By \cite{CatalyticMajorizationAndell_pNorms}, it is sufficient to prove that $||x||_p\leq ||y||_p \ \forall p\geq 1$. Now, assume for a contradiction that there exists $p\geq 1$ such that $||x||_p>||y||_p$. Write $1<\mu = \frac{||x||_p}{||y||_p}$. Then, by (\ref{products of ellp norms}), we have
    \begin{equation}
        \mu \leq \frac{||c_{i+1}||_p}{||c_i||_p} \ \ \forall i,
    \end{equation}
    and thus, by induction, $||c_n||_p\geq ||c_0||_p\mu^n \ \ \forall n$. Therefore, $||c_n||_p\rightarrow \infty$, contradicting the fact that the $\ell_p$-norms of probability vectors are bounded above by 1.

\end{proof}

Proposition \ref{LOCC closures negative result} constitutes more evidence that flexibility is probably of little use in the TT of LOCC if we assume access to arbitrary catalysts. However, it turns out that flexible catalysis can be used to strengthen catalytic PM extractions:
\begin{thm}\label{flexhelp} $\mathrm{CatExt}_{\mathrm{PM}}\subsetneq \mathrm{CatExt}_{\mathrm{PM}}^{(\mathrm{fin})}$.
\end{thm}
\begin{proof}
    Similarly to Theorem \ref{CatExtLUfin_beats_CatExtLU}, follows from Corollary \ref{anynumcat} and Remark \ref{ttofpm}.
\end{proof}

However finite flexibility does not help PM transformations:

\begin{thm} \label{flexdoesnothelp} $\mathrm{Cat}_{\mathrm{PM}}^{(\mathrm{fin})}=\mathrm{Cat}_{\mathrm{PM}}$.
\end{thm}

\begin{proof}
    Follows from Theorem \ref{abelianeq} and Remark \ref{ttofpm}.  
\end{proof}

\begin{center}

\includegraphics[width=11cm]{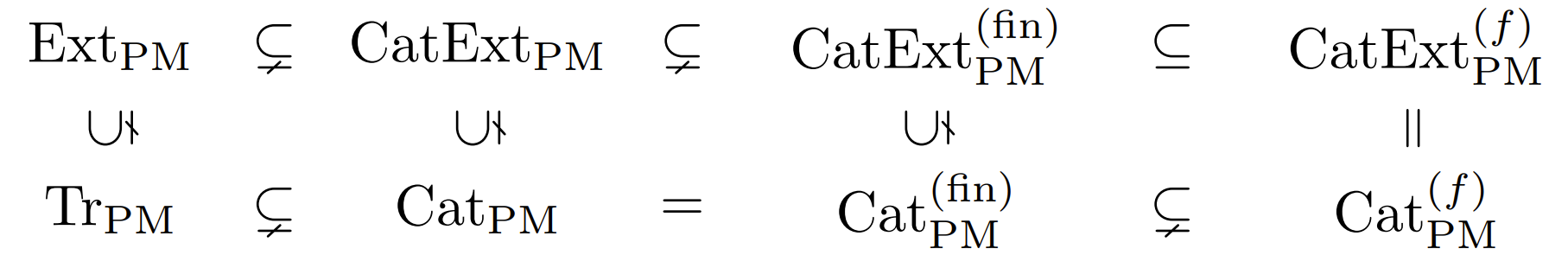}

\textbf{Figure 4}: Inclusion diagram of catalytic PM transformation classes. Note that $\mathrm{Tr}_{\mathrm{PM}}\subsetneq \mathrm{Cat}_{\mathrm{PM}}$ follows from Proposition \ref{projimp} and the fact that $\mathbb{C}^*$ has a non-trivial torsion subgroup. All the other unproved relations are trivial or follow immediately from the others.
\end{center}

\vspace{10pt}

\begin{rmk}
    Theorem \ref{flexhelp} shows that flexibility can provide an advantage in a computationally natural setting, where we are allowed to apply quantum gates corresponding to permutation matrices, and we are also allowed to discard qubits. On the other hand, Theorem \ref{flexdoesnothelp} shows that flexibility does not help if we do not have a free discard operation.
\end{rmk}

\begin{dfn}
    A bipartite quantum state is LU irreducible if it is not LU equivalent to any tensor product of two entangled bipartite quantum states.
\end{dfn}

\begin{thm}
\label{nouniqfact}
    There exist four $\mathrm{LU}$ irreducible, pairwise $\mathrm{LU}$ inequivalent bipartite quantum states $\ket{a},\ket{b},$ $\ket{c},\ket{d}$ such that $\ket{a}\ket{b}$ and $\ket{c}\ket{d}$ are $\mathrm{LU}$ equivalent.
\end{thm}
\begin{proof}
Take bipartite quantum states with the following vectors of Schmidt coefficients:
\begin{equation*}
    p_a=(1,1,\sqrt{2},\sqrt{8})/\sqrt{12}, p_b= (1,1,\sqrt{2})/2,
\end{equation*}
\begin{equation}
     p_c=(1,1,1,1,2,\sqrt{8})/4,p_d=(1,\sqrt{2})/\sqrt{3}.
\end{equation}

We will use the following fact, which is straightforward to check.

\textit{Fact. } The Schmidt vector of a tensor product of two bipartite pure states is the tensor product of the Schmidt vectors of the two states.

    In particular, the Schmidt coefficients of $\ket{a}\ket{b}$ and $\ket{c}\ket{d}$ are given by $p_a\otimes p_b$ and $p_c\otimes p_d$, respectively. $p_a\otimes p_b$ and $p_c\otimes p_d$ are the same up to permutation, so $\ket{a}\ket{b}$ and $\ket{c}\ket{d}$ are LU equivalent. It is easy to see that none of the vectors $p_a,p_b,p_c,p_d$ can be written as a non-trivial tensor product, so the quantum states $\ket{a},\ket{b},\ket{c},\ket{d}$ that they represent are LU irreducible.
\end{proof}

\begin{rmk}
    This result implies that LU equivalence classes of pure, bipartite quantum states in general do not have a unique factorisation into equivalence classes of LU irreducible states. In particular, it is possible that an irreducible state can be extracted from a tensor product of states, but not from any of the factors.
\end{rmk}

\begin{rmk}
    Problem 1.12 in \cite{kiss2025identical} deals with essentially the same question.
\end{rmk}

\section{Flexible catalysis outperforms catalytic multicopy extractions}
\label{Flexible catalysis outperforms catalytic multicopy extractions}

Proposition \ref{flexvsmulticop} (combined with Propostion \ref{Union_of_cyclic_classes}) implies that we can achieve the same transformations using a finite set of catalysts as with catalytic multicopy transformations. This does not rule out the possibility that an infinite set of catalysts can be used to perform a transformation that does not have a catalytic multicopy simulation. Although $\mathrm{Cat}_{\mathrm{LU}}^{(\mathrm{fin})}\subsetneq \mathrm{Cat}_{\mathrm{LU}}^{(f)}$, this is not that interesting because $\mathrm{Cat}_{\mathrm{LU}}^{(f)}=\mathrm{CatExt}_{\mathrm{LU}}^{(f)}$ contains extractions, while $\mathrm{Cat}_{\mathrm{LU}}^{(\mathrm{fin})}$ does not. It would be more interesting to know whether infinite flexible catalysis can outperform catalytic multicopy extractions.

 In this section, we will artificially construct TT where flexible catalysis is provably stronger than catalytic multiple copy extractions. To make the argument as simple as possible, we will work with a TT in which each equivalence class of states is a singleton. Note that this coincides with the definition of a resource theory in \cite{FRITZ_2015}: resource theories are ordered commutative monoids.

\begin{thm}
\label{flexible_advantage_for_abstract_TT}
    There exists a transformation theory $T$ in which flexible catalysis is strictly stronger than catalytic multiple copy extractions, that is, $\mathrm{CatExt}^{(\mathrm{fin})}_T\subsetneq \mathrm{CatExt}^{(f)}_T$.
\end{thm}
\begin{proof}
    We will prove the claim by explicitly constructing such an ordered commutative monoid $M$. Start from the free commutative monoid generated by distinct objects $a,b,c_1,c_2\ldots$. Let $\leq$ be the smallest partial order satisfying 
    \begin{equation}
    \label{FC assumption}
    a+c_i\leq b+c_{i+1} \ \forall i\in \mathbb{Z}_{>0}.
    \end{equation}
    
    By smallest, we mean that the only relations are those that are guaranteed by the assumption (\ref{FC assumption}), the axioms of partial orders, and translation-invariance. We will show that such a partial order $\leq$ exists and $ma+c\leq mb+c+d$ does not hold for any $c,d\in M, m\in \mathbb{Z}_{>0}$. This will be sufficient, because then for the resulting transformation theory $T$ we will have $(a,b)\in \mathrm{CatExt}^{(f)}_T\backslash \mathrm{CatExt}^{(\mathrm{fin})}_T$.

By the definition of $\leq$, we have $t\leq t'$ for some monoid elements $t,t'\in M$ if and only if this can be obtained from the relations $s\leq s, s\in M$ (reflexivity) and (\ref{FC assumption}) by finitely many applications of additive translation and sequential composition (transitivity). We claim that whenever $ua+vb+x_1c_1+\ldots+x_nc_n\leq za+wb+y_1c_1+\ldots+y_nc_n$ for some $u,v,z,w,x_i,w_i \in \mathbb{Z}_{\geq 0}$ (note that this is the general form of an inequality between two monoid elements), the following two statements must hold:

(i) for all $k\in [n]=\{1,2,\ldots,n\}$, we have $(x_i=y_i \forall i \in [k-1]) \implies x_k\geq y_k$, 

(ii) $x_i=y_i \forall i \in [n] \implies u=z, v=w$.

Firstly, observe that both (i) and (ii) are clearly satisfied for inequalities of type $s\leq s$ and (\ref{FC assumption}). Secondly, it is easy to see that these properties together are preserved under translation and sequential composition of inequalities. It follows that (i) and (ii) must hold for all inequalities $ua+vb+x_1c_1+\ldots+x_nc_n\leq za+wb+y_1c_1+\ldots+y_nc_n$ in $M$.

Now $\leq$ defines a valid partial order because it is reflexive and transitive by definition, and it is also antisymmetric because $ua+vb+x_1c_1+\ldots+x_nc_n\leq za+wb+y_1c_1+\ldots+y_nc_n$ and $ua+vb+x_1c_1+\ldots+x_nc_n\geq za+wb+y_1c_1+\ldots+y_nc_n$ imply (using (i) and (ii)) that $x_i=y_i \ \forall i, u=z,v=w$.

Finally, assume that $ma+c\leq mb+c+d$ for some $c,d\in M, m\in \mathbb{Z}_{>0}$. Write this inequality in the form $ua+vb+x_1c_1+\ldots+x_nc_n\leq za+wb+y_1c_1+\ldots+y_nc_n$ by expressing both sides in terms of the generators of $M$. We see that $x_i>y_i$ can never happen, and thus $x_i=y_i \ \forall i \in [n]$ by (i). But then $v=w$ by (ii), so the coefficient of $b$ must be equal in $ma+c$ and in $mb+c+d$, contradicting the assumption $m>0$.

\end{proof}

\begin{rmk}
    Of course, this TT was constructed artificially for the purpose of this proof, and does not have any physical meaning. However, the result is interesting, because it shows that there can be no general mathematical reason why flexible catalysis couldn't have a strict advantage over catalytic multiple copy extractions. In particular, it suggests that physically interesting resource theories might exist where flexible catalysis has this kind of advantage.
\end{rmk}

\section{Open questions}
\label{Open questions}

In this section, we will discuss some questions that are left open by our results and we deem interesting. We will start by introducing one of these questions in detail and reformulating it as a precise mathematical question (Section \ref{Infinite flexible catalysis under LU}). Then we will mention some further open directions without going into the details (Section \ref{Further open questions}).

\subsection{Infinite flexible catalysis under LU}
\label{Infinite flexible catalysis under LU}

 According to Theorem \ref{flexible_advantage_for_abstract_TT}, infinite flexible catalysis can outperform catalytic multicopy extractions if we do not insist on a physically meaningful resource theory. It is also natural to ask if the same can happen in a resource theory such as LU, where the question is equivalent to asking whether we have a strict inclusion
$\mathrm{CatExt}_{\mathrm{LU}}^{(\mathrm{fin})}\subsetneq \mathrm{CatExt}_{\mathrm{LU}}^{(f)}$, since we can perform the same transformations using a finite set of catalysts as with catalytic multicopy transformations.

\vspace{0.2cm}

\textbf{Question:}  Is $\mathrm{CatExt}_{\mathrm{LU}}^{(\mathrm{fin})}\subsetneq \mathrm{CatExt}_{\mathrm{LU}}^{(f)}$, i.e. does infinite flexible catalysis enable strictly more extractions under LU than catalytic multicopy extractions?

\vspace{0.2cm}

We will now translate this to a precise mathematical question about polynomials. We start by giving a sufficient condition for the separation.

\vspace{0.2cm}

\begin{prop}
\label{polys_would_solve_open_problem}

    Assume that there exist $p,q\in \mathbb{Z}[x]$ such that
    
    (i) $p^n \notin \mathbb{Z}_{\geq 0}[x]$ for all $n\in \mathbb{Z}_{> 0}$;

    (ii) $q\cdot p^n \in \mathbb{Z}_{\geq 0}[x]$ for all $n\in \mathbb{Z}_{\geq 0}$.

    Then $\mathrm{CatExt}_{\mathrm{LU}}^{(\mathrm{fin})}\subsetneq \mathrm{CatExt}_{\mathrm{LU}}^{(f)}$.
\end{prop}
\begin{proof}
    Let $A=\iota(pq),B=\iota(q)$, where $\iota$ is the map in Definition \ref{polys_to_multisets}. Then for any $k$ we have $(kA,kB)\notin \mathrm{CatExt}_{\mathbb{R}}$, since otherwise we would have $(pq)^k=dq^k$ for some polynomial $d\in  \mathbb{Z}_{\geq 0}[x]$, contradicting $p^k\notin  \mathbb{Z}_{\geq 0}[x]$. But if we set $C_i=\iota(q\cdot p^i)$ for all $i$, then $A+C_i=B+C_{i+1}$ for all $i$, and hence $(A,B) \in \mathrm{Cat}_\mathbb{R}^{(f)}$. Thus, by Corollary \ref{ttlueq}, $\mathrm{CatExt}_{\mathrm{LU}}^{(\mathrm{fin})}\subsetneq \mathrm{Cat}_{\mathrm{LU}}^{(f)}$.
\end{proof}

\begin{dfn}
    (Essentially positive polynomials.)
   A polynomial $p\in \mathbb{Z}[x_1,\ldots,x_m]$ is essentially positive if there exists a polynomial $c\in \mathbb{Z}[x_1,\ldots,x_m]$ such that $cp^n$ has nonnegative coefficients for all $n\geq 0$.
\end{dfn}

\begin{dfn}
(Generalisation of Definition \ref{polys to multisets}.)
    Let $\iota:\mathbb{Z}_{\geq 0}[x_1,\ldots,x_m]\backslash\{0\}\rightarrow M_{\mathbb{Z}^m}$ be the map that sends any polynomial 
     \begin{equation}
        \sum_{i_1\ldots i_m} a_{i_1,i_2,\ldots, i_m} x_1^{i_1}x_2^{i_2}\ldots x_m^{i_m}
    \end{equation}
    with nonnegative integer coefficients to the multiset whose elements come from $\{0,1,\ldots\}^m$ and the multiplicity of any $(i_1,\ldots,i_m)\in \{0,1,\ldots\}^m$ is the coefficient $a_{i_1,i_2,\ldots, i_m}$ of the polynomial. Write this multiset as
    \begin{equation}
        \{(i_1,i_2,\ldots,i_m)^{(a_{i_1,i_2,\ldots, i_m})} \}.
    \end{equation}
\end{dfn}

\begin{prop}
\label{multivariate_polys_to_multisets}
    The map $\iota:\mathbb{Z}_{\geq 0}[x_1,\ldots,x_m]\backslash\{0\}\rightarrow M_{\mathbb{Z}^m}$
    \begin{equation}
        \sum_{i_1\ldots i_m} a_{i_1,i_2,\ldots, i_m} x_1^{i_1}x_2^{i_2}\ldots x_m^{i_m} \mapsto \{(i_1,i_2,\ldots,i_m)^{(a_{i_1,i_2,\ldots, i_m})} \}
    \end{equation}
    is an injective monoid homomorphism, and its image is the set of finite $\mathbb{Z}^m$-multisets with nonnegative entries in all elements.
\end{prop}
\begin{proof}
    Analogous to the proof of Proposition \ref{polys_to_multisets}.
\end{proof}

\begin{prop}

     $\mathrm{CatExt}_{\mathrm{LU}}^{(\mathrm{fin})}\subsetneq \mathrm{CatExt}_{\mathrm{LU}}^{(f)}$ if and only if there exists $m\in \mathbb{N}$ and $p,q\in \mathbb{Z}[x_1,\ldots,x_m]$ such that
    
    (i) $p^n \notin \mathbb{Z}_{\geq 0}[x_1,\ldots,x_m]$ for all $n\in \mathbb{Z}_{> 0}$;

    (ii) $q\cdot p^n \in \mathbb{Z}_{\geq 0}[x_1,\ldots,x_m]$ for all $n\in \mathbb{Z}_{\geq 0}$.

\end{prop}

\begin{proof}
    Take $(A,B)\in  \mathrm{CatExt}_{\mathrm{LU}}^{(f)}\backslash \mathrm{CatExt}_{\mathrm{LU}}^{(\mathrm{fin})}$, and identify them with the multisets of real numbers representing them. Let $C_0,C_1,\ldots$ be a catalytic chain for $(A,B)$, i.e. $A+C_i=B+C_{i+1} \forall i\geq 0$. (Such a chain exists since $\mathrm{CatExt}_{\mathrm{LU}}^{(f)}=\mathrm{Cat}_{\mathrm{LU}}^{(f)}$ by Proposition \ref{infeqtrcat}.) Then, by finiteness of $A,B,C_0$, there is a natural number $m$ and a finitely generated group $G$ such that $\mathbb{R}\geq G\cong \mathbb{Z}^m$, and $A,B,C_0$ are $G$-multisets. Without loss of generality, we may assume that the smallest $i$\textsuperscript{th} component of a vector occurring in $A$ is $0$ (for all $1\leq i \leq m$), and similarly for $B $ and $C_0$. Then, by induction, all multisets $C_j$ in the catalytic chain are also $G$-multisets, and the smallest $i$\textsuperscript{th} component of a vector occurring in $C_j$ is $0$. Thus, $A,B,C_j$ correspond to polynomials $a,b,c_j \in \mathbb{Z}_{\geq 0}[x_1,\ldots,x_m]$ under the map defined in Proposition \ref{multivariate_polys_to_multisets}. The condition $A+C_i=B+C_{i+1} \forall i\geq 0$ translates to $ac_i=bc_{i+1} \forall i\geq 0$. Then we must have $b|a$, say $a=bp$ for some $p\in \mathbb{Z}[x_1,\ldots,x_m]$, otherwise there would exist an irreducible polynomial $q\in \mathbb{Z}[x_1,\ldots,x_m]$ that occurs at a higher power in the factorisation of $b$ than in the factorisation of $a$, which we can write as $v_q(a)<v_q(b)$. But then we would have $v_q(c_0)>v_q(c_1)>v_q(c_2)>\ldots$, leading to a contradiction.

    Now, by assumption, $(A,B)\notin \mathrm{CatExt}_{\mathrm{LU}}^{(\mathrm{fin})}$, so for all $n \in \mathbb{Z}_{>0}$ we have $p^n\notin  \mathbb{Z}_{\geq 0}[x_1,\ldots,x_m]$, otherwise $a^n=(bp)^n=b^np^n$ would describe a valid multicopy extraction. So $p$ is infinitely negative.

    However, we have $ac_i=bpc_i=bc_{i+1}$, i.e. $pc_i=c_{i+1}$ for all $i$. Setting $q=c_0$, we get $c_i=qp^n \in  \mathbb{Z}_{\geq 0}[x_1,\ldots,x_m]$. So $p$ is essentially positive.

    Conversely, given $m,p,q$ satisfying the conditions, take a group $G$ such that $\mathbb{R}\geq G\cong \mathbb{Z}^m$, and let $A,B$ be the $G$-multisets represented by the polynomials $qp$ and $q$, respectively. The same argument as in Proposition \ref{polys_would_solve_open_problem} shows that  $(A,B)\in  \mathrm{CatExt}_{\mathrm{LU}}^{(f)}\backslash \mathrm{CatExt}_{\mathrm{LU}}^{(\mathrm{fin})}$.
\end{proof}

So we can rephrase the question whether $\mathrm{CatExt}_{\mathrm{LU}}^{(\mathrm{fin})}\subsetneq \mathrm{CatExt}_{\mathrm{LU}}^{(f)}$ as follows:

\label{pqexist}
Are there polynomials $p,q\in \mathbb{Z}[\vec{x}]$ such that
    
    (i) $p^n \notin \mathbb{Z}_{\geq 0}[\vec{x}]$ for all $n\in \mathbb{Z}_{> 0}$;

    (ii) $q\cdot p^n \in \mathbb{Z}_{\geq 0}[\vec{x}]$ for all $n\in \mathbb{Z}_{\geq 0}$?

 In other words, is there a polynomial $p \in \mathbb{Z}[\vec{x}]$ that is infinitely negative, but essentially positive?
 
 We suspect that the answer is yes. For example, the polynomials $p(x,y)=1+x^3+x^2y+xy^2+y^3+x^3y-x^2y^2+xy^3+x^3y^2+x^2y^3$ with $q(x,y)=x+y$ and $p(x)=1+x+x^2-x^3+x^4-x^5+x^6+x^7+x^8$ with $q(x)=1+x+x^2+x^3$ are promising candidates.

 \subsection{Further open questions}
 \label{Further open questions}

\textbf{1. Is $\mathrm{Cat}_{\mathrm{LOCC}}\subsetneq \mathrm{Cat}_{\mathrm{LOCC}}^{(f)}$?} We have seen that flexibility with a finite number of catalysts does not increase the set of possible catalytic transformations under LOCC. So it is natural to ask whether an infinite set of catalysts can give us an advantage.

 \textbf{2. What happens if we allow mixed states?} In all three transformation theories we investigated in this paper, we only worked with pure states. It is natural to ask how our results change if the underlying sets of these transformation theories are replaced by the set of all density matrices.

\vspace{0.2cm}

\textbf{3. What happens for other classes of free quantum operations?} We have only looked at three specific classes of quantum operations. These were sufficient to show that the relative power of different catalytic transformation types depends on the class of free operations, i.e. on the resource theory. It would also be interesting to look at other classes of quantum operations, particularly those relevant for near-term quantum computation, such as the Clifford gates.

\vspace{0.2cm}

\textbf{4. Is it possible to design catalytic embeddings using flexible catalysis?} It was shown in \cite{amy2023catalytic} that a wide range of unitary gates can be efficiently simulated catalytically. Can we get any advantage by allowing flexibility in this context?

\section{Declarations}

\textbf{Acknowledgements.} The authors would like to thank Matthew Johnson, Alexandra Kowalska, and Levente Bodnar for helpful discussions. 

\textbf{Conflict of interest.} The authors declare no conflict of interest.

\textbf{Ethical statement.} This article does not contain any studies with human participants or animals performed by any of the authors.

\textbf{Informed consent.} Informed consent was not required for this study as it does not involve human participants.

\textbf{Data availability.} No datasets were generated or analysed during the current study.

\textbf{Funding.} SS acknowledges support from the Royal Society University Research Fellowship. This work was supported by the Engineering and Physical Sciences Research Council on Robust and Reliable Quantum Computing (RoaRQ) Investigation 005 [grant reference EP/W032625/1].

\bibliographystyle{ieeetr}
\bibliography{main}

 \end{document}